\documentclass[11pt]{article}

\usepackage{natbib}
\usepackage[utf8]{inputenc}
\usepackage[T1]{fontenc}
\usepackage{libertine}
\usepackage{amsmath,amsthm,amssymb}
\usepackage{algorithm}
\usepackage[noend]{algorithmic}
\usepackage{hyperref}
\usepackage{cleveref}
\usepackage{fullpage}
\usepackage{mathtools}
\usepackage{thmtools,thm-restate}
\usepackage[cmyk]{xcolor}
\usepackage{tikz}
\usetikzlibrary{positioning,arrows.meta,calc,intersections}
\allowdisplaybreaks
\usepackage{latex_abbreviations}

\DeclareMathOperator*{\argmax}{argmax}

\newcommand{\bfzero}{\mathbf{0}}
\newcommand{\bfone}{\mathbf{1}}
\newcommand{\aff}{\mathrm{aff}}

\newcommand{\nequ}{{n_{\mathrm{eq}}}}
\newcommand{\nvar}{{n_{\mathrm{var}}}}
\newcommand{\nrec}{{n_{\mathrm{rec}}}}
\newcommand{\vmin}{{v_{\mathrm{min}}}}
\newcommand{\vmax}{{v_{\mathrm{max}}}}
\newcommand{\opt}{{{\mathrm{OPT}}}}

\newtheorem{theorem}{Theorem}[]

\newtheorem{proposition}[theorem]{Proposition}
\theoremstyle{definition}
\newtheorem{definition}[theorem]{Definition}
\newtheorem{remark}[theorem]{Remark}

\title{Algorithmic Bayesian persuasion with combinatorial actions}
\author{
	Kaito Fujii\\
	National Institute of Informatics\\
	fujiik@nii.ac.jp
	\and
	Shinsaku Sakaue\\
	The University of Tokyo\\
	sakaue@mist.i.u-tokyo.ac.jp
}

\begin{document}

\maketitle

\begin{abstract}
Bayesian persuasion is a model for understanding strategic information revelation: an agent with an informational advantage, called a sender, strategically discloses information by sending signals to another agent, called a receiver. In algorithmic Bayesian persuasion, we are interested in efficiently designing the sender's signaling schemes that lead the receiver to take action in favor of the sender. This paper studies algorithmic Bayesian-persuasion settings where the receiver's feasible actions are specified by combinatorial constraints, e.g., matroids or paths in graphs. We first show that constant-factor approximation is NP-hard even in some special cases of matroids or paths. We then propose a polynomial-time algorithm for general matroids by assuming the number of states of nature to be a constant. We finally consider a relaxed notion of persuasiveness, called CCE-persuasiveness, and present a sufficient condition for polynomial-time approximability.
\end{abstract}

\section{Introduction}
Information asymmetry is ubiquitous. For example, a seller has more information about a product than a customer, and a client sometimes knows more about a delegated task than a worker. Such an agent with an informational advantage often strategically discloses information to influence the receiver's decisions. A fundamental question in information economics asks how much influence an informational advantage has or what kind of information disclosure strategy the sender should use. 

Bayesian persuasion \citep{KG11} is a model for understanding strategic information disclosure.
In the standard setting, there are two agents called the \textit{sender} and the \textit{receiver}.
The sender, who has an informational advantage, strategically reveals information by sending a signal to the receiver.
Then, according to the revealed information, the receiver selects an action that maximizes his expected utility. 
The receiver's action also affects the sender's utility, which is usually different from the receiver's utility.
The sender's goal is to design a strategy for information disclosure, called a \textit{signaling scheme}, to lead the receiver to select an action that is favorable for the sender. 
In this paper, we are interested in the algorithmic aspect of Bayesian persuasion: when can we efficiently compute (approximately) optimal signaling schemes? 

In many practical scenarios, the receiver's actions are specified by some combinatorial constraints. 
For example, a receiver's action may be a path in a network or a portfolio of a limited number of stocks. 
In such cases, lists of the receiver's actions can be prohibitively large, making standard methods for computing the sender's signaling schemes impractical. 
To reveal when we can efficiently compute optimal signaling schemes in the presence of such combinatorial actions is an important question, but it remains unexplored as mentioned in \citep[Open Question 2.8]{D17}. 
This paper addresses this open question and elucidates several classes of problems for which we can/cannot efficiently compute (approximately) optimal signaling schemes.

To motivate our work, we below give concrete situations where the receiver's actions are combinatorial. 
An important application of Bayesian persuasion is financial advice. 
Suppose the sender and the receiver to be a financial adviser and an investor, respectively. 
The sender knows accurate predictions on stock returns, which are unknown to the receiver. 
Since the sender's returns are not always aligned with those of the receiver, the sender strategically reveals information to increase her returns. 
After receiving a signal (advice) from the sender, the receiver decides which stocks to buy. 
In practice, there are a huge number of stocks, and it is unrealistic to assume that the receiver can hold a portfolio of arbitrarily many stocks. Thus, a cardinality constraint is often imposed when optimizing portfolios \citep{IHSYFKK18,ZZWLZ20}. In this situation, each receiver's action is a combination of a limited number of stocks. 

More complicated combinatorial constraints can appear in other applications. 
For example, let us imagine that a boss (sender) asks a subordinate (receiver) to create a committee by choosing a representative for each of several groups.
An action of the receiver is a set of people that contains a single person from each group.
In another example, the receiver is a company that constructs an electrical grid connecting all electricity consumers, and the sender is a government that requests the construction. 
An action of the receiver is a tree in a network that covers all electricity consumers.

\subsection{Our Results}
We study Bayesian persuasion with combinatorial actions, each of which is a subset of a finite set, denoted by $E$, and satisfies some combinatorial constraints.  
We call each component of $E$ an \textit{element}, i.e., an action is a combination of some elements. 
Utility functions of the sender and the receiver are set functions defined on $E$. 
The main issue of this setting is that there may exist exponentially many actions in $|E|$. 
We aim to clarify under what conditions we can/cannot compute (approximately) optimal signaling schemes efficiently with respect to input sizes, including $|E|$. 

First, we prove that it is NP-hard to achieve constant-factor approximation even in some special cases of matroid or path constraints.
For partition matroid constraints, we utilize a hardness result studied in public Bayesian persuasion with no externalities \citep{DX17}.
For other combinatorial constraints, including uniform matroid constraints, graphic matroid constraints, and path constraints, we construct reductions from an existing hard problem called \textsc{LINEQ-MA} \citep{GR09}, 
which asks to determine whether a large fraction of linear equations can hold or even a small fraction cannot hold.
These reductions have connections to the NP-hardness proof for the \textsc{Opt-Signal} problem \citep{CCMG20} and, as a by-product of our result on partition matroids, we obtain a hardness result for the \textsc{Opt-Signal} problem.

Next, we develop a polynomial-time algorithm for general matroid constraints by assuming the number of states of nature to be a constant (we explain what states and nature are in \Cref{sec:pre}).
We formulate the problem as an exponentially large linear program (LP) and then show that its size can be reduced by enumerating only relevant variables and constraints. 
To enumerate them efficiently, we utilize a computational-geometric algorithm. 
For some special matroids, we show that the LP size can be further reduced. 

Finally, we consider a relaxed persuasiveness condition called \textit{CCE-persuasiveness} \citep{Xu20,CCG20}, 
which is based on the concept of Bayes coarse correlated equilibria. 
We provide a sufficient condition under which a polynomial-time approximation algorithm exists. 
Specifically, if we can approximately maximize a sum of the sender's and receiver's utilities so that only the sender's has an approximation factor, we can compute approximately optimal CCE-persuasive schemes in polynomial time. 
This result is applicable to important cases where the sender's utility is a monotone submodular function.

\subsection{Related Work}
\citet{KG11} proposed the original Bayesian-persuasion model. 
Utilizing the model, researchers analyzed various social situations, including voting \citep{Sch15,AC16}, financial sector stress tests \citep{GL18}, financial markets \citep{DDZ17}, routing \citep{BCKS16}, auctions \citep{DIR14}, and information spread \citep{AB19}.
On the other hand, various algorithms have been developed for Bayesian persuasion with additional settings: exponentially many states of nature \citep{DX16}, multiple receivers \citep{AB19}, secretary problem \citep{HHS20s}, prophet inequalities \citep{HHS20p}, and payment \citep{DNPW19}.

The closest setting to ours is Bayesian persuasion with no externalities \citep{BB17,DX17,AB19,Xu20,CMCG21}.
In this setting, there are multiple receivers, and each one selects a binary action from $\{0, 1\}$ to maximize his own utility, which does not depend on the other receivers' actions. 
This setting can be decomposed into two subclasses: private and public. 
The former supposes that the sender can send a signal to each receiver privately, 
while the latter supposes that the sender sends a public signal to all the receivers. 
As we prove in \Cref{sec:hardness}, public Bayesian persuasion with no externalities is equivalent to our setting with a special partition matroid constraint. 
Our study provides hardness results for other types of constraints and algorithms for general matroid constraints, thus going beyond the existing results on public Bayesian persuasion.

Some existing studies introduced constraints on Bayesian persuasion. 
A series of studies \citep{DIR14,DKQ16,GHHS21} considered a setting where the sender can use only a limited number of bits for signaling.
\citet{BTZ21} considered ex ante and ex post constraints, which restrict structures of posterior distributions induced by the sender's signals. 
Note that the constraints considered in those studies are imposed on the sender's signaling scheme, while we consider combinatorial constraints on the receiver's possible actions.

\section{Problem Settings and Preliminaries}\label{sec:pre}
This section provides problem settings of Bayesian persuasion with combinatorial actions. 
We first describe the basic setting, where sender and receiver maximize their utility functions and the receiver's actions are represented by a matroid. 
We then present some basics and examples of matroids. 
We finally describe the setting with path constraints, where sender and receiver minimize their cost functions.

\paragraph{Notation.}
We denote the set of non-negative reals by $\bbR_{\ge 0}$.
For any set $X$, let $\Delta_X \coloneqq \{ p \colon X \to \bbR_{\ge 0} \mid \sum_{i \in X} p(i) = 1\}$ be the probability simplex over $X$.

\subsection{Basic Setting}
Let $\Theta$ be the family of states of nature. 
A state $\theta\in\Theta$ is drawn from a distribution $\mu \in \Delta_\Theta$, which is common knowledge shared by the sender and the receiver.
In the remainder of this paper, we assume $\mu(\theta) > 0$ for all $\theta \in \Theta$.
If this does not hold, we can remove such $\theta$ from the problem. 
The sender has the informational advantage of knowing the realized state $\theta$, while the receiver cannot access $\theta$ directly. 

After observing the state of nature $\theta$, the sender sends a signal $\sigma \in \Sigma$ to the receiver.
The signals can be randomized, that is, the sender can select any distribution $\phi_\theta \in \Delta_\Sigma$ for each $\theta \in \Theta$ and send signal $\sigma$ randomly sampled from $\phi_\theta$ depending on the observed state $\theta$.
A tuple $(\phi_\theta)_{\theta \in \Theta}$ of probability distributions over $\Sigma$ is called a \textit{signaling scheme}.
The core assumption of Bayesian persuasion is the \textit{commitment assumption}, which compels the sender to publicly commit to a signaling scheme before observing the state of nature $\theta$.
Under this assumption, the process goes on as follows:  
the sender publicly declares a signaling scheme $(\phi_\theta)_{\theta \in \Theta}$, 
then observes the state of nature $\theta$ sampled from $\mu$, 
and then sends signal $\sigma$ to the receiver following the distribution $\phi_\theta$.

Let $E$ be a finite set of $n$ elements. A combinatorial action is a combination of elements in $E$. 
Let $\calI\subseteq 2^E$ be the set of all possible combinatorial actions the receiver can take. 
After observing signal $\sigma$, the receiver takes a combinatorial action $S \in \calI$. 
As a result of this action, the sender and the receiver obtain utility values, which are specified by non-negative set functions $s_\theta \colon 2^E \to \bbR_{\ge 0}$ and $r_\theta \colon 2^E \to \bbR_{\ge 0}$, respectively, for each state of nature $\theta \in \Theta$.
When receiving signal $\sigma\in\Sigma$, the receiver's belief on the state of nature is represented by the posterior distribution $\xi_\sigma \in \Delta_\Theta$ such that
\begin{equation*}
	\xi_\sigma(\theta) = \frac{\mu(\theta) \phi_\theta(\sigma)}{\sum_{\theta' \in \Theta} \mu(\theta') \phi_{\theta'}(\sigma)}.
\end{equation*}
The receiver takes the best response according to this posterior distribution, that is, 
he selects $S_\sigma^*\in\calI$ such that 
\begin{equation*}
	S_\sigma^* \in \argmax_{S \in \calI} \sum_{\theta \in \Theta} \xi_\sigma(\theta) r_\theta(S).
\end{equation*}
If there are multiple best responses, we assume that ties are broken in favor of the sender.
Consequently, the sender, who commits to the signaling scheme $(\phi_\theta)_{\theta \in \Theta}$, 
obtains the following payoff in expectation: 
\begin{equation*}
	\sum_{\sigma \in \Sigma} \sum_{\theta \in \Theta} \mu(\theta) \phi_\theta(\sigma)  s_\theta(S^*_\sigma).
\end{equation*}

The revelation principle \citep{KG11} claims that to consider only direct and persuasive signaling schemes is sufficient.
A direct signaling scheme associates each signal with an action, i.e., $\Sigma = \calI$, and recommends action $S \in \calI$ by sending the corresponding signal $S \in \Sigma$.
A direct signaling scheme is \textit{persuasive} if the receiver has no incentive to deviate from recommendation $S \in \calI$ when receiving signal $S \in \Sigma$. 
The persuasiveness constraint requires signaling scheme $(\phi_\theta)_{\theta\in\Theta}$ to satisfy $\sum_{\theta \in \Theta} \xi_S(\theta) r_\theta(S) \ge \sum_{\theta \in \Theta} \xi_S(\theta) r_\theta(S')$ for every pair $S, S' \in \calI$, where $\xi_S$ is the posterior distribution when the receiver observes $S\in\Sigma$.

By considering the revelation principle, we can formulate the problem of computing an optimal signaling scheme as
\begin{equation}\label{eq:lp}
	\begin{alignedat}{4}
	&\text{maximize} &\ \ &\sum_{\theta \in \Theta} \mu(\theta) \sum_{S \in \calI} \phi_\theta(S) s_\theta(S)\\
	&\text{subject to} & &\sum_{\theta \in \Theta} \mu(\theta) \phi_\theta(S) \left( r_\theta(S) - r_\theta(S') \right) \ge 0& \quad & (S, S' \in \calI)\\
	& & &\phi_\theta \in \Delta_\calI & \quad & (\theta \in \Theta).
	\end{alignedat}
\end{equation}
The first constraint is the persuasiveness constraint, under which $S$ must be one of the receiver's best responses when the sender sends signal $S$.
The second constraint requires $\phi_\theta$ to be a probability distribution in $\Delta_\calI$ for each $\theta\in\Theta$.

In general, the utility set functions $(s_\theta)_{\theta \in \Theta}$ and $(r_\theta)_{\theta \in \Theta}$ do not have a polynomial-size representation in $|E|$ and $|\Theta|$.
In such cases, we assume that we have value oracles, i.e., 
$s_\theta(S)$ and $r_\theta(S)$ values are available for any given $\theta \in \Theta$ and $S \subseteq E$.
Similarly, if the set $\calI$ of the receiver's actions does not have a polynomial-size representation, we assume access to an independence oracle, which returns whether $S \in \calI$ or not for any given $S \subseteq E$. 
These assumptions are common in combinatorial optimization. 

For the hardness results in \Cref{sec:hardness}, we let $(s_\theta)_{\theta \in \Theta}$ and $(r_\theta)_{\theta \in \Theta}$ be linear functions, i.e., $s_\theta(S) = \sum_{i \in S} s_\theta(\{i\})$ and $r_\theta(S) = \sum_{i \in S} r_\theta(\{i\})$ for all $S \subseteq E$ and $\theta \in \Theta$. 
Such $(s_\theta)_{\theta \in \Theta}$ and $(r_\theta)_{\theta \in \Theta}$ have polynomial-size representations. 
We also use sets $\calI$ of actions that have polynomial-size representations. 
Thus, our hardness results indeed come from computational hardness, not from the hardness of representing problems. 
For the polynomial-time algorithms in \Cref{sec:exp}, we assume $(r_\theta)_{\theta \in \Theta}$ to be linear but allow $(s_\theta)_{\theta \in \Theta}$ to be a general set function.
For the CCE-persuasiveness result in \Cref{sec:cce}, we allow both $(s_\theta)_{\theta \in \Theta}$ and $(r_\theta)_{\theta \in \Theta}$ to be general set functions, while we impose an approximability assumption as detailed later. 

\subsection{Matroids}
Many of our results consider combinatorial constraints represented by matroids, 
which are useful to model various combinatorial actions that appear in practice. 

A pair $(E, \calI)$ of a finite set $E$ and a non-empty set family $\calI \subseteq 2^E$ is called a matroid if the following conditions hold:
\begin{itemize}
	\item $S \subseteq T \in \calI$ implies $S \in \calI$.
	\item For any $S, T \in \calI$ with $|S| < |T|$, there exists $i \in T \setminus S$ such that $S \cup \{i\} \in \calI$.
\end{itemize}
A set $S \subseteq E$ is called independent if $S \in \calI$ holds.
Given a matroid $(E, \calI)$ and a non-negative weight $w(i) \in \bbR_{\ge 0}$ for each element $i \in E$, 
a maximum weight independent set $S^* \in \argmax_{S \in \calI} \sum_{i \in S} w(i)$ can be found by using the greedy algorithm, 
which starts with the empty set and, in descending order of weights, adds elements that maintain independence.
Below, we introduce some special matroids that are useful in practical scenarios.

\paragraph{Uniform matroids.}
A uniform matroid is a matroid with $\calI = \{ S \subseteq E \mid |S| \le k \}$ for some integer $k > 0$.
Under a uniform matroid constraint, the receiver selects at most $k$ elements that yield the largest expected utility value.
This constraint fits the portfolio optimization scenario, where the receiver selects $k$ stocks expected to yield the largest return. 

\paragraph{Partition matroids.}
Let $E_1, \dots, E_P \subseteq E$ be a partition, i.e., $E_1 \cup \cdots \cup E_P = E$ and $E_i \cap E_j = \emptyset$ for every distinct $i,j \in [P]$. 
Assign some positive integer $k_i$ for each $i \in [P]$. 
A matroid with $\calI = \{ S \subseteq E \mid \forall i \in [P], ~ |S \cap E_i| \le k_i \}$ 
is called a partition matroid. 
Under a partition matroid constraint, the receiver selects at most $k_i$ elements from the $i$th partition.
This can model a scenario where the receiver creates a committee by choosing one person from each group.

\paragraph{Graphic matroids.}
Given an undirected graph $(V, E)$, a graphic matroid $(E, \calI)$ is a matroid with $\calI = \{ S \subseteq E \mid \text{$S$ does not contain a cycle}\}$.
Any maximal independent set of a graphic matroid forms a spanning tree (or a spanning forest if graph $(V, E)$ is not connected).
If $E$ is a set of all edges that can be used for constructing an electrical grid and $V$ is a set of all electricity consumers, 
a graphic matroid constraint models the situation where the receiver constructs an electrical grid covering all consumers.

\subsection{Minimization Setting with Path Constraints}
Given a directed graph $G = (V, E)$ with origin $v_s \in V$ and destination $v_t \in V$, 
we consider a setting where the receiver selects a $v_s$--$v_t$ path $S \subseteq E$ according to a signal from the sender. 
This setting appears, for example, when the sender is an association that manages traffic by recommending a route, and the receiver is a taxi driver.
In this setting, the sender and the receiver usually aim to minimize their costs rather than maximize utility functions. 
Thus, we focus on the minimization setting regarding path constraints 
(we distinguish maximization and minimization settings since some of our results 
concerning approximability cannot be translated from one to the other). 
The problem of computing an optimal signaling scheme can be formulated as a minimization version of \eqref{eq:lp}, 
where $\calI$ is the set of all $v_s$--$v_t$ paths. 
The receiver selects a path that minimizes his expected cost based on the posterior distribution; 
therefore, we reverse the inequality sign of the persuasiveness constraint. 
The sender's goal is to minimize her expected cost.

\section{Hardness Results}\label{sec:hardness}
We show the NP-hardness of obtaining a constant-factor approximation for Bayesian persuasion 
with partition matroid, uniform matroid, graphic matroid, and path constraints. 
Throughout this section, we assume the utility functions of the sender and the receiver to be linear. 

\subsection{Partition Matroid Constraints}\label{subsec:partition-matroids-hardness}
We show that Bayesian persuasion with a partition matroid constraint is intractable even in a special case where each partition has two elements. We prove this via a reduction from public Bayesian persuasion with no externalities.
\citet{DX17} proved that constant-factor approximation for public Bayesian persuasion with no externalities is NP-hard even if we restrict the sender's utility function to $s_\theta(S) = |S|$ for each $\theta \in \Theta$.
By associating each partition with a receiver of public Bayesian persuasion, we obtain the reduction, thus proving the following hardness result.

\begin{restatable}{theorem}{thmhardnesspartition}\label{thm:hardness-partition}
	For any constant $\alpha \in (0, 1]$, it is NP-hard to compute an $\alpha$-approximate solution for Bayesian persuasion with a partition matroid constraint whose partitions are of size two.
\end{restatable}

The proof and the details of public Bayesian persuasion with no externalities are provided in \Cref{sec:hardness-partition}.

\begin{remark}
	We can also construct a reduction for the opposite direction. Therefore, we establish an equivalence in terms of the approximability between public Bayesian persuasion with no externalities and Bayesian persuasion with a partition matroid constraint whose partitions are of size two.
\end{remark}

As a by-product of \Cref{thm:hardness-partition}, we can obtain a hardness result on the \textsc{Opt-Signal} problem \citep{CCMG20}, which is a variant of Bayesian persuasion where the receiver's utility depends on a random type that is unknown to the sender.
We provide the details in \Cref{sec:hardness-optsignal}.

\subsection{Uniform Matroid Constraints}
We prove the NP-hardness of constant-factor approximation for Bayesian persuasion with a uniform matroid constraint.
Our proof is inspired by that of the \textsc{Opt-Signal} problem \citep{CCMG20}, 
which is equivalent to the special case of partition matroid constraints, as mentioned above. 
Note that a uniform matroid cannot be represented by the special partition matroid with size-two partitions. 
Nevertheless, we can use a similar proof strategy to the one used for the \textsc{Opt-Signal} problem.
Specifically, we construct a reduction from a variant of an existing hard problem called \textsc{LIENQ-MA} \citep{GR09}, which asks us to distinguish whether a linear system $A x = c$ has a solution satisfying most of the equations or no solution satisfies even a small fraction of the equations.
The definition of this problem and the proof of the following theorem are provided in \Cref{sec:hardness-uniform}.

\begin{restatable}{theorem}{thmhardnessuniform}\label{thm:hardness-uniform}
	For any constant $\alpha \in (0, 1]$, it is NP-hard to compute an $\alpha$-approximate solution for Bayesian persuasion with a uniform matroid constraint.
\end{restatable}

\subsection{Graphic Matroid Constraints}
The hardness result for Bayesian persuasion with graphic matroid constraints can be proved in a similar way to that of uniform matroid constraints, i.e., we construct a reduction from the variant of the \textsc{LINEQ-MA} problem.
We provide the proof in \Cref{sec:hardness-graphic}.

\begin{restatable}{theorem}{thmhardnessgraphic}\label{thm:hardness-graphic}
	For any constant $\alpha \in (0, 1]$, it is NP-hard to compute an $\alpha$-approximate solution for Bayesian persuasion with a graphic matroid constraint.
\end{restatable}

\subsection{Path Constraints}
Recall that in this setting, the sender and the receiver aim to minimize their expected costs, 
while the above settings consider maximizing utilities. 
Thus, the approximation ratio for this problem is lower-bounded by $1$, and a smaller value implies a better approximation. 
The proof is similar to those for uniform or graphic matroids (see \Cref{sec:hardness-path}).

\begin{restatable}{theorem}{thmhardnesspath}\label{thm:hardness-path}
	For any constant $\alpha \in [1, \infty)$, it is NP-hard to compute an $\alpha$-approximate solution for Bayesian persuasion with a path constraint.
\end{restatable}

\section{Polynomial-time Algorithm for Constant Number of States}\label{sec:exp}
We present a polynomial-time algorithm for Bayesian persuasion with general matroid constraints by assuming the number of states of nature to be a constant, i.e., $|\Theta| = O(1)$.
We denote by $d$ the dimension of $\Delta_\Theta$, i.e., $d \coloneqq |\Theta| - 1$.
Throughout this section, we assume the receiver's utility function to be linear, while the sender's is a general set function, 
and we regard computation costs that depend only on $d$ as constants. 

\begin{figure}[tb]
	\centering
	\begin{tikzpicture}[line width=1pt]
		\definecolor{color1}{cmyk}{0,0,1.00,0.52}
		\definecolor{color2}{cmyk}{1.00,0.50,0,0.19}
		\definecolor{color3}{cmyk}{0,0.66,0.66,0.40}
	
		\newlength{\ul}
		\setlength{\ul}{0.3\hsize}
		\coordinate (v1) at (0.0\ul, 0.0\ul);
		\coordinate (v2) at (1.0\ul, 0.0\ul);
		\coordinate (v3) at (0.5\ul, 0.8660254\ul);
		\coordinate (c) at (0.5\ul, 0.3\ul);
		\coordinate (r11) at ([xshift=+0.5\ul, yshift=+0.2\ul]c);
		\coordinate (r12) at ([xshift=-0.6\ul, yshift=-0.24\ul]c);
		\coordinate (r1c) at ($0.95*(r11)+0.05*(r12)$);
		\coordinate (r21) at ([xshift=+0.6\ul, yshift=-0.30\ul]c);
		\coordinate (r22) at ([xshift=-0.5\ul, yshift=+0.25\ul]c);
		\coordinate (r2c) at ($0.05*(r21)+0.95*(r22)$);
		\coordinate (r31) at ([xshift=+0.07\ul, yshift=+0.56\ul]c);
		\coordinate (r32) at ([xshift=-0.065\ul, yshift=-0.52\ul]c);
		\coordinate (r3c) at ($0.05*(r31)+0.95*(r32)$);
		\path [name path=v1--v2] (v1) -- (v2);
		\path [name path=v2--v3] (v2) -- (v3);
		\path [name path=v3--v1] (v3) -- (v1);
		\path [name path=line1] (r11) -- (r12);
		\path [name path=line2] (r21) -- (r22);
		\path [name path=line3] (r31) -- (r32);
	
		\path [name intersections={of=line1 and v2--v3, by=i1}];
		\path [name intersections={of=line2 and v3--v1, by=i2}];
		\path [name intersections={of=line3 and v1--v2, by=i3}];
		\path [line width=5pt, draw=white, fill=color1, opacity=1] (c) -- (i1) -- (v2) --  (i3) --cycle;
		\path [line width=5pt, draw=white, fill=color2, opacity=1] (c) -- (i2) -- (v1) --  (i3) --cycle;
		\path [line width=5pt, draw=white, fill=color3, opacity=1] (c) -- (i1) -- (v3) --  (i2) --cycle;
	
		\draw (v1) -- (v2) -- (v3) -- cycle;
		\draw (r11) -- (r12);
		\draw (r21) -- (r22);
		\draw (r31) -- (r32);
		\draw [-{Stealth}] (r1c) -- ([xshift=+0.05\ul, yshift=-0.10\ul]r1c);
		\draw [-{Stealth}] (r1c) -- ([xshift=-0.05\ul, yshift=+0.10\ul]r1c);
		\node [anchor=north, inner sep=0pt] at ([xshift=+0.25\ul, yshift=-0.125\ul]r1c) {\fontsize{8}{8}\selectfont$\bbE r_\theta(\{1\}) > \bbE r_\theta(\{2\})$};
		\node [anchor=south, inner sep=0pt] at ([xshift=+0.10\ul, yshift=+0.125\ul]r1c) {\fontsize{8}{8}\selectfont$\bbE r_\theta(\{1\}) < \bbE r_\theta(\{2\})$};
		\draw [-{Stealth}] (r2c) -- ([xshift=+0.05\ul, yshift=+0.10\ul]r2c);
		\draw [-{Stealth}] (r2c) -- ([xshift=-0.05\ul, yshift=-0.10\ul]r2c);
		\node [anchor=south, inner sep=0pt] at ([xshift=-0.10\ul, yshift=+0.125\ul]r2c) {\fontsize{8}{8}\selectfont$\bbE r_\theta(\{1\}) < \bbE r_\theta(\{3\})$};
		\node [anchor=north, inner sep=0pt] at ([xshift=-0.20\ul, yshift=-0.125\ul]r2c) {\fontsize{8}{8}\selectfont$\bbE r_\theta(\{1\}) > \bbE r_\theta(\{3\})$};
		\draw [-{Stealth}] (r3c) -- ([xshift=+0.1\ul, yshift=-0.0125\ul]r3c);
		\draw [-{Stealth}] (r3c) -- ([xshift=-0.1\ul, yshift=+0.0125\ul]r3c);
		\node [anchor=west, inner sep=0pt] at ([xshift=+0.125\ul, yshift=-0.020\ul]r3c) {\fontsize{8}{8}\selectfont$\bbE r_\theta(\{2\}) < \bbE r_\theta(\{3\})$};
		\node [anchor=east, inner sep=0pt] at ([xshift=-0.125\ul, yshift=-0.005\ul]r3c) {\fontsize{8}{8}\selectfont$\bbE r_\theta(\{2\}) > \bbE r_\theta(\{3\})$};
		\node [inner sep=1pt, anchor=north east] at (v1) {$\theta_3$};
		\node [inner sep=1pt, anchor=north west] at (v2) {$\theta_2$};
		\node [inner sep=1pt, anchor=south] at (v3) {$\theta_1$};
	
		\node [align=center, text=white] at ($0.25*(c)+0.25*(i1)+0.25*(v2)+0.25*(i3)$) {$\{1,3\}$ is\\the best.};
		\node [align=center, text=white] at ($0.25*(c)+0.25*(i2)+0.25*(v1)+0.25*(i3)$) {$\{1,2\}$ is\\the best.};
		\node [align=center, text=white] at ($0.25*(c)+0.25*(i1)+0.25*(v3)+0.25*(i2)$) {$\{2,3\}$ is\\the best.};
	\end{tikzpicture}
	\caption{%
		An illustration of cells for Bayesian persuasion with a uniform matroid constraint, where the receiver selects the top-$2$ elements from $E = \{1,2,3\}$.
		For each pair of distinct $i, j \in E$, a hyperplane (a line) divides the probability simplex $\Delta_\Theta$ into the region where $\bbE_{\theta \sim \xi} [r_\theta(\{i\})] > \bbE_{\theta \sim \xi} [r_\theta(\{j\})]$ holds and the complement.
		The receiver's best response is identical within each cell.
	}\label{fig:matroid-simplex}
	\end{figure}
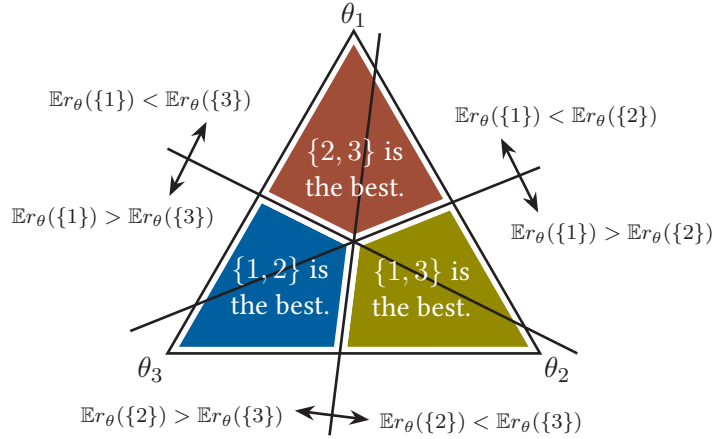

The main obstacle is the fact that the number of variables and constraints in LP \eqref{eq:lp} is exponential in $n \coloneqq |E|$. 
In the case of matroid constants, however, we can show that only a small number of variables can take non-zero values and that a large fraction of constraints is unnecessary. 
Thus, if we can efficiently enumerate all relevant variables and constraints, we can obtain a smaller LP that is equivalent to the original LP \eqref{eq:lp} and can be solved more cheaply. 
This approach is inspired by the one studied in public Bayesian persuasion with no externalities \citep{Xu20}, 
which is equivalent to Bayesian persuasion with a special partition matroid constraint, as mentioned \Cref{subsec:partition-matroids-hardness}.  
Our algorithm can be seen as an extension of \citep{Xu20} to the case of general matroids. 
Below we present the technical details. 

A key observation is that if $d$ is small, the number of combinatorial actions that the receiver can take is also small.
We define the set $\calI^*$ of all combinatorial actions that can be the best response for some posterior distribution as
\begin{equation*}
	\calI^* = \bigg\{ S \in \calI ~\bigg|~ \exists \xi \in \Delta_\Theta \colon S \in \argmax_{S \in \calI} \sum_{\theta \in \Theta} \xi(\theta) r_\theta(S)  \bigg\}.
\end{equation*}
By using this reduced set of combinatorial actions, we formulate a smaller LP as follows: 
\begin{equation}\label{eq:reduced-lp}
\begin{alignedat}{4}
&\text{maximize} &\ \ &\sum_{\theta \in \Theta} \mu(\theta) \sum_{S \in \calI^*} \phi_\theta(S) s_\theta(S)\\
&\text{subject to} & &\sum_{\theta \in \Theta} \mu(\theta) \phi_\theta(S) \left( r_\theta(S) - r_\theta(S') \right) \ge 0& \quad & (S, S' \in \calI^*)\\
& & &\phi_\theta \in \Delta_\calI & \quad & (\theta \in \Theta).
\end{alignedat}
\end{equation}
This LP formulation has variables $\phi_\theta(S)$ only for $S \in \calI^*$, 
and the first constraint exists only for $S, S' \in \calI^*$.
We can show that this reduced LP formulation \eqref{eq:reduced-lp} is equivalent to the original one \eqref{eq:lp}.
The proof is provided in \Cref{sec:reduced-lp}.

\begin{restatable}{proposition}{propreducedlp}\label{prop:reduced-lp}
	There is a bijection between the feasible regions of \eqref{eq:lp} and \eqref{eq:reduced-lp} that does not change the objective value.
\end{restatable}

As illustrated in \citep[Example 1]{Xu20}, if the receiver's expected utility has certain degeneracy, the receiver may have exponentially many best responses, i.e., $|\calI^*| = \Omega(2^n)$.
To exclude such troubling corner cases, we make a mild non-degeneracy assumption regarding the receiver's utility.

\begin{restatable}{assumption}{asmdegeneracy}\label{asm:degeneracy}
	For each $i \in E$, let $\psi_i = (r_\theta(\{i\}))_{\theta \in \Theta} \in \bbR^\Theta$ be a vector representing the utilities of the $i$th element. 
	We assume that, for any permutation $\pi \colon \{1,\dots,n\} \to E$ and any subset $S \subseteq \{1,\dots,n-1\}$ with $|S| = |\Theta|$, vectors $(\psi_{\pi(i)} - \psi_{\pi(i+1)})_{i \in S}$ are linearly independent. 
\end{restatable}

Under this assumption, we develop a polynomial-time algorithm, which first enumerates all combinatorial actions in $\calI^*$ and then solves the reduced LP \eqref{eq:reduced-lp}.
We below describe our algorithm for general matroids and then present faster algorithms for some special matroids.

\subsection{General Matroids}
We explain how to enumerate $\calI^*\subseteq\calI$ in polynomial time when $(E, \calI)$ is a matroid. 
Let $\aff(\Delta_\Theta) = \{ \xi \colon \Theta \to \bbR \mid \sum_{\theta \in \Theta} \xi(\theta) = 1 \}$ be the affine hull of $\Delta_\Theta$, 
i.e., the smallest affine space containing $\Delta_\Theta$. 
We consider hyperplane $h_{ij} = \{ \xi \in \aff(\Delta_\Theta) \mid \sum_{\theta \in \Theta} \xi(\theta) \left( r_\theta(\{i\}) - r_\theta(\{j\}) \right) = 0 \}$ for each pair of distinct $i,j \in E$. 
Each hyperplane $h_{ij}$ divides $\aff(\Delta_\Theta)$ into two halfspaces: the region where the $i$'s expected utility $\bbE_{\theta \sim \xi} [r_\theta(\{i\})] = \sum_{\theta \in \Theta} \xi(\theta) r_\theta(\{i\})$ is at least that of $j$ and the complement (see \Cref{fig:matroid-simplex}).
Thus, the set of hyperplanes $\calH = \{ h_{ij} \mid i,j \in E, ~ i \neq j\}$ divides $\aff(\Delta_\Theta)$ into small pieces, which we call \textit{cells}. 
Within each cell, the descending order of expected utility values is identical. 

\Cref{alg:general-matroid} presents the details of our algorithm.
We first enumerate all cells and then use the greedy algorithm to obtain a maximum weight independent set for an interior point of each cell. 
Although $\calH$ has degeneracy, by using an algorithm for constructing arrangements of hyperplanes \citep{Ede87}, we can enumerate all cells.
Finally, we solve the reduced LP \eqref{eq:reduced-lp} where $\calI^*$ 
is replaced with the family of obtained independent sets, denoted by $\calI_\calC$.

\begin{algorithm}[tb]
\caption{Algorithm for Bayesian persuasion with a general matroid constraint}
\label{alg:general-matroid}
\textbf{Input}: value oracles of $s_\theta$ and $r_\theta$, independence oracle of $\calI$, $\mu$.\\
\textbf{Output}: signaling scheme $(\phi_\theta)_{\theta \in \Theta}$.\par
\begin{algorithmic}[1]
\STATE Let $\calH = \{ h_{ij} \mid i,j \in E, ~ i \neq j \}$ be the set of hyperplanes, where $h_{ij} = \{ \xi \in \aff(\Delta_\Theta) \mid \sum_{\theta \in \Theta} \xi(\theta) \left( r_\theta(\{i\}) - r_\theta(\{j\}) \right) = 0 \}$.
\STATE Obtain the set of all cells $\calC$ of the arrangement of $\calH$ by using cell enumeration algorithm \citep{Ede87}.
\STATE $\calI_\calC\gets\emptyset$.
\FOR{each cell $C \in \calC$}
	\STATE Let $\xi \in C$ be any interior point of $C$.
	\STATE Apply the greedy algorithm to matroid $(E, \calI)$ and weights $\{\bbE_{\theta \sim \xi} [r_\theta(\{i\})]\}_{i\in E}$ to obtain the maximum weight independent set and add it to $\calI_\calC$.
\ENDFOR
\STATE Solve LP \eqref{eq:reduced-lp} with $\calI^* = \calI_\calC$ and obtain solution $(\phi_\theta)_{\theta \in \Theta}$.
\STATE \textbf{return} signaling scheme $(\phi_\theta)_{\theta \in \Theta}$.
\end{algorithmic}
\end{algorithm}

Now, we show that $\calI_\calC$ obtained by the algorithm recovers $\calI^*$, the family of possible best responses. 
The following lemma guarantees that considering only interior points of cells is sufficient for enumerating $\calI^*$ under \Cref{asm:degeneracy}, i.e.,  we do not need to care about the boundaries of cells. 

\begin{restatable}{lemma}{lemdegeneracy}\label{lem:degeneracy}
	Let $\psi_1, \dots, \psi_n \in \bbR^\Theta$ be vectors that satisfy \Cref{asm:degeneracy}.
	Then, for any permutation $\pi$, $\{ \xi \in \bbR^\Theta \mid \psi_{\pi(1)}^\top \xi \ge \cdots \ge \psi_{\pi(n)}^\top \xi ~\text{and}~ \bfone^\top \xi = 1\} \neq \emptyset$ implies $\{ \xi \in \bbR^\Theta \mid \psi_{\pi(1)}^\top \xi > \cdots > \psi_{\pi(n)}^\top \xi  ~\text{and}~ \bfone^\top \xi = 1\} \neq \emptyset$.
\end{restatable}

Therefore, we can enumerate all possible best responses $\calI^*$ by collecting maximum weight independent sets returned by the greedy algorithm, whose behavior depends only on the descending order of the weights. 

\begin{restatable}{theorem}{lemgeneralenumeration}\label{lem:general-enumeration}
	Under \Cref{asm:degeneracy}, $\calI^* \subseteq \calI_\calC$ holds.
	Moreover, $|\calI_\calC| = O(n^{2d})$.
\end{restatable}

The proofs of the above lemma and theorem are provided in \Cref{sec:degeneracy}.
Since $|\calH| = O(n^2)$, the cells can be enumerated in $O(n^{2d})$ time \citep{Ede87}. 
The resulting LP \eqref{eq:reduced-lp} has $O(n^{2d})$ variables and $O(n^{4d})$ constraints, which can be solved in $\textrm{poly}(n^d)$ time.

\subsection{Faster Algorithms for Special Matroids}
As described above, our algorithm first enumerates every combinatorial action that can be the receiver's best response for some $\xi \in \Delta_\Theta$, and then solves the reduced LP. 
For several special cases of matroids, we can use faster enumeration algorithms and obtain better upper bounds on $|\calI^*|$. 

\paragraph{Uniform matroids.}
In this case, a receiver's action consists of top-$k$ elements, and thus we can enumerate all possible best responses $\calI^*$ in a simpler manner. 
We consider $n$ hyperplanes on $\Delta_\Theta$ associated with $n$ expected utility values, 
and enumerate all combinations of $k$ hyperplanes that correspond to top-$k$ elements for some $\xi \in \Delta_\Theta$. 
The enumeration problem is closely related to a discrete-geometric subject, called the $k$-level in an arrangement of hyperplanes. 
By using an algorithm of \citep{Mul91} that enumerates all faces of level at most $k$, we can enumerate all $k$-level faces corresponding to $\calI^*$ in $O(k^{\lceil (d+1)/2 \rceil} n^{\lfloor (d+1)/2 \rfloor})$ time for $d \ge 3$, and in $O(k^{\lceil (d+1)/2 \rceil} n^{\lfloor (d+1)/2 \rfloor} \log (n/k))$ time for $d \le 2$.
Moreover, \citet{CS89} presented an upper bound on the number of faces of level at most $k$, which implies $|\calI^*| = O(k^{\lceil (d+1)/2 \rceil} n^{\lfloor (d+1)/2 \rfloor})$.

\paragraph{Partition matroids.}
As with the case of general matroids, we enumerate cells of the hyperplane arrangement, but the number of hyperplanes can be reduced.
A key observation is that the maximum weight independent set is determined by the order of weights of elements in each partition, not by the order of weights of all elements.
Therefore, it is sufficient to enumerate all possible orders of weights in each partition.
Let $n_1 \coloneqq |E_1|, n_2 \coloneqq |E_2|,\dots, n_P \coloneqq |E_P|$ be the size of each partition. 
For each $p \in [P]$, 
we consider hyperplanes for all pair of distinct $i,j \in E_p$. 
As with the case of general matroids, we can obtain $\calI_\calC\supseteq\calI^*$ by enumerating the cells of the arrangement of those hyperplanes.
If $n_1 = \cdots = n_P$, the number of hyperplanes is $O(n^2 / P)$.
Therefore, in this case, we have $|\calI^*| = O(n^{2d} / P^d)$.

\paragraph{Graphic matroids.}
In this case, we need to enumerate all spanning trees $S \in \calI$ that can attain the maximum weight for some $\xi \in \Delta_\Theta$. 
To this end, we can use existing results on the parametric spanning tree problem. 
When $d = 1$, by using an algorithm of \citep{FSE96}, we can enumerate all the spanning trees $\calI^*$ in $O(|V||E| \log |V|)$ time.
Moreover, when $d = 1$, it is known that $|\calI^*| = O(|E||V|^{1/3})$ holds \citep{Dey98}.

\begin{remark}
	For path constraints, even if $d = 1$, the number of paths that can be the shortest path for some $\xi \in \Delta_\Theta$ is known to be $n^{\Omega(\log n)}$ in general \citep{Car83}.
	Thus, our approach, which enumerates every path that can be the shortest path for some $\xi \in \Delta_\Theta$, fails to run in polynomial time even if $d$ is a constant.
	We need a different approach for obtaining an efficient algorithm for path constraints.
\end{remark}

\section{Polynomial-time Algorithm for CCE-Persuasiveness}\label{sec:cce}
We present how to compute approximately optimal CCE-persuasive signaling schemes for Bayesian persuasion with combinatorial actions. 
To motivate us to study CCE-persuasiveness, a relaxed notion of persuasiveness, in the combinatorial setting, let us consider the following investment example. 
Let the receiver be an investor wondering whether to buy an investment trust or build a portfolio by himself. 
If the investor decides to buy an investment trust, stocks are bought following a policy of a portfolio manager, who is the sender.  
The sender wants to sell the investment trust, and a signaling scheme corresponds to her portfolio-management policy. 
The receiver buys the investment trust if it is at least as beneficial as the portfolio built following his prior belief. 
In this example, in contrast to the standard setting, 
the sender's recommendation is regarded as persuasive if it is better than the best action taken based on the receiver's prior belief. 
This idea of persuasiveness can be modeled by CCE-persuasiveness, as detailed below.

As pointed out by \citet{BM16b}, 
Bayesian persuasion is the problem of finding a Bayes correlated equilibrium that is optimal for the sender. 
By relaxing the condition of Bayes correlated equilibria, we can obtain a broader class of equilibria called Bayes coarse correlated equilibria, as mentioned in several studies \citep{CP14,HST15}.
In a Bayes coarse correlated equilibrium, the receiver has no incentive to ignore the signal, 
i.e., following the signal is at least as beneficial as the best action selected based on his prior distribution. 
Let $C = \max_{S \in \calI} \sum_{\theta \in \Theta} \mu(\theta) r_\theta(S)$ be 
the receiver's utility when he selects the best action based on his prior distribution. 
The LP for computing an optimal CCE-persuasive signaling scheme can be written as follows: 
\begin{equation}\label{eq:cce-lp}
	\begin{alignedat}{4}
		&\text{maximize} &\ \ &\sum_{\theta \in \Theta} \sum_{S\in \calI} \mu(\theta) \phi_\theta(S) s_\theta(S)\\
		&\text{subject to} & &\sum_{\theta \in \Theta} \sum_{S\in \calI} \mu(\theta) \phi_\theta(S) r_\theta(S)   \ge  C & \quad &\\
		& & &\phi_\theta \in \Delta_\calI & \quad & (\theta \in \Theta). 
	\end{alignedat}
\end{equation}

We present a sufficient condition for achieving a constant-factor approximation for the above LP. 
For public Bayesian persuasion, the equivalence between exact maximization of certain set functions and computation of CCE-persuasive schemes has been established in \citep[Theorem 5.1]{Xu20}. 
Compared with the existing result, 
ours is restricted to one direction (maximization implies persuasion). 
However, it is applicable to important cases where maximization can be done only approximately. 
For example, we can allow the sender's utility to be a monotone submodular function, as explained later. 
To prove the result, we carefully combine an existing binary search framework for approximate separation oracles \citep{Jan03,JMS03} and the proof of \citep[Theorem 5.1]{Xu20}. 
Throughout this section, we do not make any assumption on the sender's and the receiver's utility functions; 
instead, we assume that an approximate maximization oracle, together with upper and lower bounds on the optimal value, are available. 

\begin{algorithm}[tb]
\caption{Algorithm for CCE-persuasiveness}
\label{alg:cce}
\textbf{Input}: value oracles of $s_\theta$ and $r_\theta$, independence oracle of $\calI$, $\alpha$-approximation oracle, $\mu$, $\vmin$, $\vmax$, $\epsilon$.\\
\textbf{Output}: signaling scheme $(\phi_\theta)_{\theta \in \Theta}$.\par
\begin{algorithmic}[1]
\STATE $v_l \gets 0$, $v_u \gets \vmax$, $v \gets \frac{v_l + v_u}{2}$.
\STATE $\epsilon' \gets \epsilon \vmin$.
\FOR{$t = 1,2,\dots,\lceil \log(\vmax/\epsilon') \rceil + 1$}
	\STATE Apply the ellipsoid method with an $\alpha$-approximate separation oracle to dual LP \eqref{eq:cce-dual} with additional constraint $-C \cdot y + \sum_{\theta \in \Theta} x_\theta \le v$.
	\IF{$v$ is approximately feasible}
		\STATE $v_u \gets v$ and $v \gets \frac{v_l + v_u}{2}$.
	\ELSE [$v$ is infeasible]
		\STATE $v_l \gets v$ and $v \gets \frac{v_l + v_u}{2}$.
	\ENDIF
\ENDFOR
\IF{$v_l = 0$}
	\STATE \textbf{return} any signaling scheme.
\ENDIF
\STATE Solve restricted primal LP to obtain solution $(\phi_\theta)_{\theta \in \Theta}$.
\STATE \textbf{return} signaling scheme $(\phi_\theta)_{\theta \in \Theta}$.
\end{algorithmic}
\end{algorithm}

\begin{theorem}\label{thm:cce-polytime}
	Let $\alpha\in(0,1]$. 
	For any $y\ge0$ and $\theta\in\Theta$, 
	assume that a polynomial-time $\alpha$-approximation oracle that returns $S\in\calI$ with the following guarantee is available: $s_\theta(S) + y \cdot r_\theta(S) \ge \alpha \cdot s_\theta(S') + y \cdot r_\theta(S')$ for any $S'\in\calI$. 
	Moreover, assume $\vmin$ and $\vmax$ to be given as inputs such that (i) $\opt > 0$ implies $\opt > \vmin$ and (ii) $\opt < \vmax$, where $\opt$ is the optimal value.
	Then, for any $\epsilon\in(0, \alpha)$, there is a polynomial-time algorithm that computes an $(\alpha - \epsilon)$-approximate CCE-persuasive signaling scheme. 
\end{theorem}

\begin{proof}
	The dual of the LP \eqref{eq:cce-lp} can be written with variables $\{x_\theta\}_{\theta\in\Theta}$ and $y \in \bbR$ as follows: 
	\begin{equation}\label{eq:cce-dual}
	\begin{alignedat}{3}
		&\text{minimize} &\ \ & -C \cdot y + \sum_{\theta \in \Theta} x_\theta\\
		&\text{subject to} & & x_\theta - \mu(\theta) r_\theta(S) \cdot y \ge \mu(\theta) s_\theta(S) & \quad & (\theta \in \Theta, S \in \calI)\\
		& & & y\ge0.
	\end{alignedat}
	\end{equation}
	We consider applying the ellipsoid method to this LP. 
	Here, since we cannot directly use the $\alpha$-approximation oracle as a separation oracle, 
	we employ a binary-search framework of \citep{Jan03,JMS03}, which enables us to combine the ellipsoid method with the following $\alpha$-approximate separation oracle. 
	The proof of the following lemma is provided in \Cref{sec:cce-separation}.

	\begin{restatable}{lemma}{lemseparation}\label{lem:approximate_separation}
		Under the assumption of \Cref{thm:cce-polytime}, 
		there is a polynomial-time $\alpha$-approximate separation oracle 
		such that for any $\{x_\theta\}_{\theta\in\Theta}$ and $y$, 
		it either returns a separating hyperplane or guarantees the feasibility of 
		$\{x_\theta / \alpha\}_{\theta\in\Theta}$ and $y/\alpha$. 
	\end{restatable}

	We perform a binary search on $[0, \vmax]$ to estimate the optimal value of the dual LP (see \Cref{alg:cce}). 
	Given an estimated value $v\in\bbR$, 
	we add $-C \cdot y + \sum_{\theta \in \Theta} x_\theta \le v$ to the constraints in \eqref{eq:cce-dual} 
	and check the feasibility of the resulting inequality system using the ellipsoid method with the $\alpha$-approximate separation oracle. 
	Note that since we can use only an $\alpha$-approximate separation oracle, the ellipsoid method guarantees that the estimated value $v$ is either approximately feasible or infeasible.
	The binary search continues until the interval becomes smaller than $\epsilon' = \epsilon \vmin$.
	Let $v^*$ be the smallest approximately feasible value found by this binary search, and let $(\{x^*_\theta\}_{\theta\in\Theta}, y^*)$ be its corresponding solution. 
	From \Cref{lem:approximate_separation}, 
	$(\{x^*_\theta/\alpha\}_{\theta\in\Theta}, y^*/\alpha)$ satisfies the constraints of the dual LP \eqref{eq:cce-dual}, 
	and thus the optimal value is at most $v^*/\alpha$. 
	Furthermore, since $v^*-\epsilon'$ is infeasible, the optimal value of the dual LP \eqref{eq:cce-dual} is in $(v^*-\epsilon',v^*/\alpha]$. 

	If the binary search does not find any infeasible $v$, then $v^* \in [0, \epsilon \vmin)$ and the optimal value is less than $\vmin$.
	From the assumption on $\vmin$, the optimal value is $0$. Thus, any signaling scheme is optimal.
	Otherwise, to obtain a solution, we construct a restricted version of the primal LP \eqref{eq:cce-lp} as follows. 
	When executing the ellipsoid method for the smallest $v$ such that $v \ge v^*-\epsilon'$, we check polynomially many constraints, which are enough to conclude that the optimal value of the dual LP is larger than $v^*-\epsilon'$. 
	This process yields a restricted dual LP with the polynomially many constraints, 
	each of which corresponds to a primal variable. 
	By allowing only those primal variables to be non-zero, we can obtain a restricted primal LP. 
	From the strong duality, the optimal value of the restricted primal LP is also larger than $v^*-\epsilon'$. 
	On the other hand, the weak duality implies that the optimal value of the (non-restricted) primal LP \eqref{eq:cce-lp} is at most $v^*/\alpha$. 
	Thus, by solving the restricted primal LP, we can compute an $(\alpha-\epsilon)$-approximate solution in polynomial time.  
\end{proof}

For several classes of utility functions, we can implement approximate separation oracles that run in polynomial time.
For example, 
we can exactly solve $\max_{S\in\calI} s_\theta(S) + y\cdot r_\theta(S)$ if both $s_\theta$ and $r_\theta$ are linear and 
$\calI$ is an independence system of a matroid. 
When $s_\theta$ is a monotone submodular function, $r_\theta$ is linear, and $\calI$ is an independence system of a matroid, 
the continuous greedy algorithm can be used as a $(1-1/\mathrm{e})$-approximation oracle \citep[Theorem 1]{SVW17}. 
Furthermore, when $s_\theta$ is a monotone submodular function, $r_\theta$ is a gross substitute function, and $\calI$ is an independence system of a uniform matroid, 
a variant of the continuous greedy algorithm serves as a $(1-1/\mathrm{e})$-approximation oracle \citep{SY18}. 
By considering the minimization version, we can use a shortest-path algorithm as an exact separation oracle if both $s_\theta$ and $r_\theta$ are linear and $\calI$ is the set of paths on some graph.
We detail how to set $\vmin$ and $\vmax$ in \Cref{sec:cce-parameters}.

\section*{Acknowledgements}
The first author is thankful to Metteo Castiglioni for providing an explanation on the difference between the original and binary-vector versions of the \textsc{LINEQ-MA} problem.
The authors thank the anonymous reviewers for their valuable comments.
This work was supported by JST, ACT-X Grant Number JPMJAX190P, Japan.

\bibliographystyle{plainnat}
\bibliography{main}

\appendix
\section{Further Related Work}
\citet{CCM021} studied a signaling problem in a Bayesian network congestion game, in which each of multiple receivers selects a path on a network whose edges are associated with uncertain costs.
They developed a polynomial-time algorithm that computes an optimal signaling scheme under CCE-persuasiveness constraints.
Their setting is similar to ours when we consider path constraints, but there are several differences.
In their setting, there are multiple receivers whose actions affect other receivers' costs through congestion, while in our setting, there is only a single receiver.
Moreover, they assume that the sender's objective is to minimize the social cost (the sum of the receivers' costs), while we do not make such an assumption on the sender's objective.

\citet{KMP14} and \citet{MSS20} considered a similar algorithmic framework for signaling based on the multi-armed bandit problem.
In their setting, an informed principal sends a signal to each of multiple myopic receivers who arrive one by one, and aims at maximizing the social welfare.
Their analyses are based on the framework of regret minimization, and theoretically different from Bayesian persuasion.

\section{Omitted Details of Hardness Results}

\subsection{Proof for Partition Matroid Constraints}\label{sec:hardness-partition}
We prove the hardness result for Bayesian persuasion with a partition matroid constraint by constructing a reduction from public Bayesian persuasion with no externalities.

We describe the problem setting of public Bayesian persuasion with no externalities, which was introduced by \citet{AB19}.
There are $\nrec$ receivers, each of whom takes a binary action from $\{0, 1\}$.
The $i$th receiver obtains a payoff $\tilde{r}_\theta(i, 1) \in \bbR_{\ge 0}$ when he takes action $1$ and $\tilde{r}_\theta(i, 0) \in \bbR_{\ge 0}$ otherwise.
From the assumption of no externalities, these values do not depend on the other receivers' actions.
The sender's utility function $\tilde{s}_\theta \colon 2^{[\nrec]} \to \bbR_{\ge 0}$ is a set function, whose value depends on the set $S \subseteq [\nrec] \coloneqq \{1,\dots,\nrec\}$ of receivers who take action $1$.
We here consider the case where $\tilde{s}_\theta$ is linear.
The sender and the receivers share the same prior distribution $\mu$ on $\theta \in \Theta$.
The sender sends a common signal $\sigma \in \Sigma$ to all the receivers, and then the receivers take an action based on the common posterior distribution.
This problem can be formulated as
\begin{equation}\label{eq:public}
	\begin{alignedat}{4}
	&\text{maximize} &\ \ &\sum_{\theta \in \Theta} \mu(\theta) \sum_{S \subseteq [\nrec]} \phi_\theta(S) \tilde{s}_\theta(S)\\
	&\text{subject to} & &\sum_{\theta \in \Theta} \mu(\theta) \phi_\theta(S) \tilde{r}_\theta(i) \ge 0 & \quad & (S \subseteq [\nrec] ~ \text{with} ~ i \in S)\\
	& & &\sum_{\theta \in \Theta} \mu(\theta) \phi_\theta(S) \tilde{r}_\theta(i) \le 0 & \quad & (S \subseteq [\nrec] ~ \text{with} ~ i \not\in S)\\
	& & &\phi_\theta \in \Delta_{2^{[\nrec]}} & \quad & (\theta \in \Theta),
	\end{alignedat}
\end{equation}
where $\tilde{r}_\theta(i) = \tilde{r}_\theta(i, 1) - \tilde{r}_\theta(i, 0)$.

\citet{DX17} proved that constant-factor approximation for public Bayesian persuasion with no externalities is NP-hard even if we restrict the sender's utility function to $s_\theta(S) = |S|$ for each $\theta \in \Theta$.

\begin{theorem}[Theorem 6.2 in \citep{DX17}]\label{thm:hardness-public}
	For any constant $\alpha \in (0, 1]$, it is NP-hard to compute an $\alpha$-approximate solution for public Bayesian persuasion with no externalities even if the sender's utility function is $s_\theta(S) = |S|$ for each $\theta \in \Theta$.
\end{theorem}

We use this result to prove our hardness result for Bayesian persuasion with partition matroid constraints. 
By regarding each receiver of public Bayesian persuasion as a partition, we can construct the reduction and obtain the following hardness result.

\thmhardnesspartition*

\begin{proof}
	Suppose that a problem instance of public Bayesian persuasion with no externalities is given.
	We construct a problem instance of Bayesian persuasion with a partition matroid constraint whose partitions are of size two.
	We use the given set $\Theta$ of states of nature and the prior distribution $\mu$.
	Let $E_i = \{a_{i,0}, a_{i,1}\}$ be a partition for each $i \in [\nrec]$, and define the ground set $E = \bigcup_{i \in [\nrec]} E_i$.
	Let $k_1 = \cdots = k_\nrec = 1$, i.e., $\calI = \{ S \subseteq E \mid | S \cap E_i | \le 1 \}$.
	Given the $i$th receiver's utility $\tilde{r}_\theta(i, 0) \in \bbR_{\ge 0}$ and $\tilde{r}_\theta(i, 1) \in \bbR_{\ge 0}$ for each $i \in [\nrec]$, we define the receiver's utility as $r_\theta(\{a_{i,0}\}) = \tilde{r}_\theta(i, 0)$ and $r_\theta(\{a_{i,1}\}) = \tilde{r}_\theta(i, 1)$.
	Given the sender's utility function $\tilde{s}_\theta \colon 2^{[\nrec]} \to \bbR_{\ge 0}$ for each $\theta \in \Theta$, we define the sender's utility for Bayesian persuasion with the partition matroid constraint as $s_\theta(S) = \tilde{s}_\theta(\{i \in [\nrec] \mid a_{i,1} \in S\})$ for each $\theta \in \Theta$ and $S \subseteq E$.

	We show that for any signaling scheme $(\phi_\theta)_{\theta \in \Theta}$, the sender's expected utility values with respect to $(\tilde{s}_\theta)_{\theta \in \Theta}$ and $(s_\theta)_{\theta \in \Theta}$ are identical.
	Let $\xi_\sigma \in \Delta_\Theta$ be the posterior distribution when signal $\sigma \in \Sigma$ is sent. 
	It is sufficient to show that the sender's expected utility values are identical for each of the posterior distributions $\{\xi_\sigma\}_{\sigma \in \Sigma}$.
	Fix $\sigma \in \Sigma$.
	In the problem instance of public Bayesian persuasion, the $i$th receiver selects $1$ if $\sum_{\theta \in \Theta} \xi_\sigma(\theta) \tilde{r}_\theta(i, 0) \le \sum_{\theta \in \Theta} \xi_\sigma(\theta) \tilde{r}_\theta(i, 1)$ and selects $0$ otherwise.
	If we let $\tilde{S}_\sigma \subseteq [\nrec]$ be the set of the receivers who takes action $1$, then the sender's expected utility value for $\xi_\sigma$ is $\sum_{\theta} \xi_\sigma(\theta) \tilde{s}_\theta(\tilde{S}_\sigma)$.
	On the other hand, in the problem instance of Bayesian persuasion with a partition matroid constraint, for each partition $E_i$ ($i \in [\nrec]$), the receiver selects $a_{i,1}$ if $\sum_{\theta \in \Theta} \xi(\theta) r(a_{i,0}) \le \sum_{\theta \in \Theta} \xi(\theta) r_\theta(a_{i,1})$ and $a_{i,0}$ otherwise.
	Thus, if $S_\sigma \subseteq E$ is the receiver's best response, the sender's expected utility value for $\xi_\sigma$ is $\sum_{\theta} \xi_\sigma(\theta) s_\theta(S_\sigma) = \sum_{\theta} \xi_\sigma(\theta) \tilde{s}_\theta(\{i \in [\nrec] \mid a_{i,1} \in S_\sigma\}) = \sum_{\theta} \xi_\sigma(\theta) \tilde{s}_\theta(\tilde{S}_\sigma)$, which coincides with the one for public Bayesian persuasion. 
	Therefore, \Cref{thm:hardness-public} implies the hardness of achieving a constant-factor approximation for Bayesian persuasion with partition matroid constraints. 
\end{proof}

\subsection{Hardness Result for \textsc{Opt-Signal} Problem}\label{sec:hardness-optsignal}
We give an interesting connection between Bayesian persuasion with a partition matroid constraint and the \textsc{Opt-Signal} problem \citep{CCMG20}.
The \textsc{Opt-Signal} problem is a variant of Bayesian persuasion where the receiver's utility depends on a random type. 
Specifically, 
let $\calT$ be the set of possible types of the receiver and $\calA$ be the set of the receiver's actions.
We define $\tilde{s}_\theta(t, a) \in \bbR_{\ge 0}$ and $\tilde{r}_\theta(t, a) \in \bbR_{\ge 0}$ as the sender's and the receiver's utility functions, respectively, when the receiver takes action $a \in \calA$ and his type is $t \in \calT$.
The sender constructs a signaling scheme $(\phi_\theta)_{\theta \in \Theta}$, which determines a probability distribution over signals $\Sigma$ for each $\theta \in \Theta$.
The receiver observes the type $t \in \calT$, as well as the signal $\sigma$ from the sender, and then takes an action based on them.
The receiver's type conforms to the uniform distribution $\rho \in \Delta_\calT$, i.e., $\rho(t) = \frac{1}{|\calT|}$ for all $t \in \calT$. 
Our goal is to find a signaling scheme $(\phi_\theta)_{\theta \in \Theta}$ that maximizes the sender's expected utility. 

\citet{CCMG20} proved that additive approximation for \textsc{Opt-Signal} is NP-hard.
Their proof uses a problem instance of \textsc{Opt-Signal} with three actions, which yields the following result.

\begin{theorem}[Theorem 1 of \citep{CCMG20}]
	Suppose the sender's utility lies within $[0, 1]$, i.e., $\tilde{s}_\theta(t, a) \in [0, 1]$ for all $t \in \calT$, $a \in \calA$, and $\theta \in \Theta$.
	For any constant $\alpha \in (0, 1]$, it is NP-hard to compute an $\alpha$-additive approximate solution for \textsc{Opt-Signal} even if the number of actions is three.
\end{theorem}

Below we prove a hardness result, which is incomparable to their result.
Specifically, we prove the NP-hardness of constant-factor approximation for \textsc{Opt-Signal} with two actions by providing an approximation-preserving reduction from Bayesian persuasion with a partition matroid constraint whose partitions are of size two.

\begin{theorem}\label{thm:hardness-optsignal}
	For any constant $\alpha \in (0, 1]$, it is NP-hard to compute an $\alpha$-approximate solution for \textsc{Opt-Signal} even if the number of actions is two.
\end{theorem}

\begin{proof}
	We construct an approximation-preserving reduction from Bayesian persuasion with a partition matroid constraint whose partitions are of size two to \textsc{Opt-Signal} with two actions.

	Given an instance of Bayesian persuasion with a partition matroid constraint, we construct an instance of \textsc{Opt-Signal} as follows.
  If there exists a partition from which the receiver can select both of the two elements, then he always selects both elements since the utility of each element is non-negative. 
  Thus, we can assume that the receiver can select only one element from each partition without loss of generality. 
  Suppose there are $P$ partitions, and let $E_t = \{(t, 0), (t, 1)\}$ be the $t$th partition for each $t \in [P]$.
  Let $\calA = \{0,1\}$ be the set of actions for the \textsc{Opt-Signal} instance, and $\calT = [P]$ be the set of types.
  The set $\Theta$ of the states of nature and the prior distribution $\mu \in \Delta_\Theta$ are identical to those of the partition-matroid instance.
  The distribution over types is the uniform distribution $\rho(t) = \frac{1}{P}$ for all $t \in [P]$.
  Define the receiver's utility function $\tilde{r}_\theta \colon \calT \times \calA \to \bbR$ as $\tilde{r}_\theta(t, a) = |\calT| r_\theta(\{(t,a)\})$.
  Similarly, define the sender's utility function $\tilde{s}_\theta \colon \calT \times \calA \to \bbR$ as $\tilde{s}_\theta(t, a) = |\calT| s_\theta(\{(t,a)\})$.

	We consider applying the same signaling scheme to the partition-matroid instance and the \textsc{Opt-Signal} instance.
	We show that the sender's expected utility of any signaling scheme in the partition-matroid instance is equal to that in the \textsc{Opt-Signal} instance.
	Let $(\phi_\theta)_{\theta \in \Theta}$ be a signaling scheme.

	First, we consider the sender's expected utility in the \textsc{Opt-Signal} instance.
	Define the expected utility of action $a \in \calA$ for the receiver with type $t \in \calT$ when signal $\sigma \in \Sigma$ is received as  
	\begin{equation*}
		\tilde{r}_\sigma(t, a) = \frac{\sum_{\theta \in \Theta} \mu(\theta) \phi_\theta(\sigma) \tilde{r}_\theta(t, a)}{\sum_{\theta \in \Theta} \mu(\theta) \phi_\theta(\sigma)}.
	\end{equation*}
	Let $a_\sigma^t \in \argmax_{a \in \calA} \tilde{r}_\sigma(t, a)$ be one of the best responses to signal $\sigma \in \Sigma$ for the receiver with type $t \in \calT$.
	Then, the sender's expected utility is $\frac{1}{|\calT|} \sum_{t \in \calT} \sum_{\sigma \in \Sigma} \sum_{\theta \in \Theta} \mu(\theta) \phi_\theta(\sigma) \tilde{s}_\theta(t, a_\sigma^t)$.

	Next, we consider the sender's expected utility in the partition-matroid instance.
	For each signal $\sigma \in \Sigma$, the receiver's expected utility of taking action $S \subseteq E$ is
	\begin{equation*}
		\frac{\sum_{\theta \in \Theta} \mu(\theta) \phi_\theta(\sigma) r_\theta(S)}{\sum_{\theta \in \Theta} \mu(\theta) \phi_\theta(\sigma)} = \frac{1}{|\calT|} \sum_{(t, a) \in S} \frac{\sum_{\theta \in \Theta} \mu(\theta) \phi_\theta(\sigma) \tilde{r}_\theta(t, a)}{\sum_{\theta \in \Theta} \mu(\theta) \phi_\theta(\sigma)} = \frac{1}{|\calT|} \sum_{(t, a) \in S} \tilde{r}_\sigma(t, a).
	\end{equation*}
	Let $S_\sigma^* \in \calI$ be the receiver's best response that maximizes this expected utility.
	From the definition of a partition matroid constraint, $(t, a) \in S^*_\sigma$ is chosen independently for each $t \in \calT$. 
	Therefore, the receiver's best response for each $t \in \calT$ is $\argmax_{a \in \calA} \tilde{r}_\sigma(t, a)$ (ties are broken in favor of the sender). 
	Thus, there exists a bijection between the receiver's best response in the partition-matroid instance and those in the \textsc{Opt-Signal} instance, i.e., $S^*_\sigma = \{ (t, a_\sigma^t) \mid t \in \calT, ~ a_\sigma^t \in \argmax_{a \in \calA} \tilde{r}_\sigma(t, a) \}$.
	The sender's expected utility when the receiver takes action $S^*_\sigma$ for signal $\sigma \in \Sigma$ is
	\begin{align*}
		\sum_{\sigma \in \Sigma} \sum_{\theta \in \Theta} \mu(\theta) \phi_\theta(\sigma) s_\theta(S^*_\sigma)
		&= \sum_{\sigma \in \Sigma} \sum_{\theta \in \Theta} \mu(\theta) \phi_\theta(\sigma) \left( \frac{1}{|\calT|} \sum_{(t, a) \in S^*_\sigma} \tilde{s}_\theta(t, a) \right)\\
		&= \frac{1}{|\calT|} \sum_{\theta \in \Theta}  \sum_{\sigma \in \Sigma} \sum_{t \in \calT} \mu(\theta) \phi_\theta(\sigma) \tilde{s}_\theta(t, a_\sigma^t),
	\end{align*}
	which is equal to the sender's expected utility in the \textsc{Opt-Signal} instance. 
	Therefore, \Cref{thm:hardness-partition} implies the NP-hardness of achieving a constant-factor approximation for \textsc{Opt-Signal} with two actions. 
\end{proof}

\subsection{Proof for Uniform Matroid Constraints}\label{sec:hardness-uniform}

We prove the NP-hardness of constant-factor approximation for Bayesian persuasion with a uniform matroid constraint.
Our proof uses almost the same reduction as the one used for proving the NP-hardness of constant-factor approximation for the \textsc{Opt-Signal} problem.

\citet{CCMG20} proved the hardness of the \textsc{Opt-Signal} problem by reducing a binary-vector version of the \textsc{LIENQ-MA} problem to the \textsc{Opt-Signal} problem.

\begin{definition}[Binary-vector version of $\textsc{LINEQ-MA}(1-\zeta, \delta)$]
	Suppose $0 \le \delta \le 1 - \zeta \le 1$.
	Given linear equations $A x = c$ with $A \in \bbQ^{\nequ \times \nvar}$ and $c \in \bbQ^{\nequ}$, $\textsc{LINEQ-MA}(1-\zeta, \delta)$ is the promise problem of distinguishing between the following two cases: 
	\begin{itemize}
		\item There exists $x \in \{0,1\}^\nvar$ that satisfies at least a $1 - \zeta$ fraction of the equations.
		\item Every $x \in \bbQ^\nvar$ satisfies less than a $\delta$ fraction of the equations.
	\end{itemize}
\end{definition}

In the original definition of $\textsc{LINEQ-MA}(1-\zeta, \delta)$, the former case also considers $x \in \bbQ^\nvar$ instead of $x \in \{0,1\}^\nvar$.
However, the proof for the original \textsc{LINEQ-MA} by \citet{GR09} shows that the binary-vector version is also NP-hard. 

\begin{theorem}[Theorem 5.1 of \citep{GR09}]
	For any $0 \le \delta \le 1 - \zeta \le 1$, the binary-vector version of $\textsc{LINEQ-MA}(1-\zeta, \delta)$ is NP-hard.
\end{theorem}

Here, we provide a reduction from the binary-vector version of $\textsc{LINEQ-MA}(1-\zeta, \delta)$ to Bayesian persuasion with a uniform matroid constraint, thus proving its hardness. 

\thmhardnessuniform*

\begin{proof}
	Let $\zeta, \delta \in [0, 1]$ be constants such that $\alpha(1 - 2\zeta) > \delta$.
	We provide a reduction from $\textsc{LINEQ-MA}(1-\zeta,\delta)$ to $\alpha$-approximation to Bayesian persuasion with a uniform matroid constraint.  
	In other words, we show that if an $\alpha$-approximation algorithm for Bayesian persuasion with a uniform matroid constraint is available, we can distinguish the two cases of $\textsc{LINEQ-MA}(1-\zeta,\delta)$ by checking outputs of the $\alpha$-approximation algorithm.  
	This reduction is almost the same as the one for \textsc{Opt-Signal} given by \citet{CCMG20}, but we provide a full description for completeness.

	\paragraph{Reduction.}
	Given $A$ and $c$, we normalize them so that $\bar{A} = \frac{1}{\tau} A$ and $\bar{c} = \frac{1}{\tau^2} c$, where the normalization constant is $\tau = 2 \max \{ \max_{i,j} |A_{ij}|, \max_{i} |c_i|, \nvar^2 \}$.
	Let $\Theta = \{0,1,2,\dots,\nvar\}$ be the set of all possible states of nature.
	Let $\mu \in \Delta_\Theta$ be a prior distribution defined as $\mu(\theta) = \frac{1}{\nvar^2}$ for each $\theta \in [\nvar]$ and $\mu(0) = \frac{\nvar-1}{\nvar}$.
	Let $E = [\nequ] \times \{ a_0, a_1, a_2 \}$ be the ground set.
	We consider a uniform matroid constraint $S \in \calI = \{ S \subseteq E \mid |S| \le \nequ \}$.
	We define the receiver's utility $r_\theta(\{(t,a)\})$ of element $(t,a) \in E$ as
	\begin{align*}
		& r_0(\{(t,a_0)\}) = 1, & & r_0(\{(t,a_1)\}) = 1 + \bar{c}_t, & & r_0(\{(t,a_2)\}) = 1 - \bar{c}_t,\\
		& r_\theta(\{(t,a_0)\}) = 1, & & r_\theta(\{(t,a_1)\}) = 1 - \bar{A}_{t \theta} + \bar{c}_t, & & r_\theta(\{(t,a_2)\}) = 1 + \bar{A}_{t \theta}- \bar{c}_t.
	\end{align*}
	for each $\theta \in [\nvar]$.
	From the definition of $\tau$, these weights are all non-negative.
	Using these weights, we define the receiver's utility as $r_{\theta}(S) = \sum_{(t,a) \in S} r_\theta(\{(t,a)\})$.
	Define the sender's utility as $s_\theta(S) = \frac{1}{\nequ} |\{(t,a) \in S \mid a = a_0\}|$. 

	\paragraph{Completeness.}
	Suppose there exists a vector $\hat{x} \in \{0,1\}^\nvar$ such that at least a $1-\zeta$ fraction of the equations in $A \hat{x} = c$ are satisfied.
  Let $\bar{x} = \frac{1}{\tau} \hat{x}$.
  If $(A \hat{x})_t = c_t$, then $(\bar{A} \bar{x})_t = \bar{c}_t$.
  We show that the optimal expected utility is greater than $1-2\zeta$, implying that an $\alpha$-approximation algorithm for Bayesian persuasion with a uniform matroid constraint outputs a signaling scheme with an expected utility greater than that $\alpha(1-2\zeta)$.

	We consider the following signaling scheme $(\phi_\theta)_{\theta \in \Theta}$ with signals $\Sigma = \{\sigma_1, \sigma_2\}$.
	For each $\theta \in [\nvar]$, we define $\phi_{\theta}(\sigma_1) = q \bar{x}_\theta$ and $\phi_{\theta}(\sigma_2) = 1 - \phi_{\theta}(\sigma_1)$, where $q = \frac{\nvar(\nvar-1)}{1 - \sum_{i=1}^\nvar \bar{x}_i}$. 
	For $\theta = 0$, we define $\phi_0(\sigma_1) = 1$ and $\phi_0(\sigma_2) = 0$.
	Then, the posterior distribution $\xi_{\sigma_1}$ for signal $\sigma_1$ is
	\begin{equation*}
		\xi_{\sigma_1}(\theta) = \frac{\frac{q}{\nvar^2} \bar{x}_\theta}{\sum_{i=1}^\nvar \frac{q}{\nvar^2} \bar{x}_i + \frac{\nvar-1}{\nvar}} = \frac{\frac{\nvar-1}{\nvar} \frac{1}{1-\sum_{i=1}^\nvar \bar{x}_i} \bar{x}_\theta}{\frac{\nvar-1}{\nvar} \frac{1}{1-\sum_{i=1}^\nvar \bar{x}_i}\sum_{i=1}^\nvar \bar{x}_i + \frac{\nvar-1}{\nvar}} = \bar{x}_\theta
	\end{equation*}
	for each $\theta \in [\nvar]$ and $\xi_{\sigma_1}(0) = 1 - \sum_{i=1}^\nvar \bar{x}_i$.
	When $\sigma_1$ is observed, the expected weight of each element $(t,a) \in E$ for the receiver can be calculated as
	\begin{align*}
		\sum_{\theta \in \Theta} \xi_{\sigma_1}(\theta) r_{\theta}(\{(t,a_0)\}) &= 1,\\
		\sum_{\theta \in \Theta} \xi_{\sigma_1}(\theta) r_{\theta}(\{(t,a_1)\}) &= \sum_{\theta=1}^\nvar \bar{x}_\theta \left(1 - \bar{A}_{t\theta} + \bar{c}_t\right) + \left(1 - \sum_{\theta=1}^\nvar \bar{x}_\theta \right) \left(1 + \bar{c}_t\right) = 1 + \bar{c}_t - (\bar{A} \bar{x})_t,\\
		\sum_{\theta \in \Theta} \xi_{\sigma_1}(\theta) r_{\theta}(\{(t,a_2)\}) &= \sum_{\theta=1}^\nvar \bar{x}_\theta \left(1 + \bar{A}_{t\theta} - \bar{c}_t\right) + \left(1 - \sum_{\theta=1}^\nvar \bar{x}_\theta \right) \left(1 - \bar{c}_t\right) = 1 - \bar{c}_t + (\bar{A} \bar{x})_t.
	\end{align*}
	For $t \in [\nequ]$ with $(\bar{A} \bar{x})_t = \bar{c}_t$, the receiver's expected utility of $(t,a)$ is $1$ for each $a \in \{a_0,a_1,a_2\}$.
	For other $t \in [\nequ]$, while the receiver's expected utility of $(t,a_0)$ is $1$, 
	either $(t,a_1)$ or $(t,a_2)$ has a utility value greater than $1$ and the other has a value smaller than $1$. 
	Under the uniform matroid constraint, the receiver selects the elements with the top-$\nequ$ expected utility values, 
	and at most $\zeta \nequ$ elements have expected utilities greater than $1$.
	Since we can break ties in favor of the sender, when $\sigma_1$ is observed, the receiver selects $S \subseteq E$ that contains at most $\zeta \nequ$ elements $(t, a)$ such that $a \neq a_0$, 
	i.e., at least $(1 - \zeta) \nequ$ elements $(t, a) \in S$ satisfy $a = a_0$. 
	Therefore, the expected utility of the sender is at least
	\begin{equation*}
		\left( \sum_{\theta \in \Theta} \mu(\theta) \phi_{\theta}(\sigma_1) \right) \frac{1}{\nequ} (1-\zeta)\nequ \ge \frac{\nvar-1}{\nvar} (1-\zeta) > 1-2\zeta,
	\end{equation*}
	where the second inequality is obtained by letting $\nvar$ be sufficiently large.
	Therefore, an $\alpha$-approximation algorithm returns a signaling scheme with an expected utility greater than $\alpha(1-2\zeta)$.

	\paragraph{Soundness.}
	We prove that if every vector $\hat{x} \in \bbQ^\nvar$ satisfies less than a $\delta$ fraction of the equations $A \hat{x} = c$, no algorithm for Bayesian persuasion with a uniform matroid constraint returns a signaling scheme with an expected utility more than or equal to $\delta$.
	Suppose for contradiction that there exists a signaling scheme $(\phi_\theta)_{\theta \in \Theta}$ that provides the sender with an expected utility at least $\delta$.
	We show there exists a vector $\bar{x} \in \bbQ^\nvar$ satisfying at least a $\delta$ fraction of the equations $\bar{A} \bar{x} = \bar{c}$.

	By an averaging argument, there exists a signal $\sigma \in \Sigma$ such that the receiver's best response to $\sigma$ yields the sender's expected utility that is at least $\delta$.
	Let $S^* \subseteq E$ be this best response to $\sigma$, which contains elements with the top-$\nequ$ expected utility values.
	From the definition of the sender's utility, $S^*$ contains at least $\delta \nequ$ elements $(t, a)$ such that $a = a_0$. 
	Let $\xi_\sigma \in \Delta_\Theta$ be the posterior probability distribution when $\sigma$ is observed.
	The receiver's expected utility value for each element is
	\begin{align*}
		\sum_{\theta \in \Theta} \xi_{\sigma}(\theta) r_{\theta}(\{(t,a_0)\}) &= 1,\\
		\sum_{\theta \in \Theta} \xi_{\sigma}(\theta) r_{\theta}(\{(t,a_1)\}) &= \sum_{\theta \in \Theta} \xi_\sigma(\theta) \left(1 - \bar{A}_{t\theta} + \bar{c}_t\right) = 1 + \bar{c}_t - \sum_{\theta=1}^\nvar \xi_\sigma(\theta) \bar{A}_{t \theta},\\
		\sum_{\theta \in \Theta} \xi_{\sigma}(\theta) r_{\theta}(\{(t,a_2)\}) &= \sum_{\theta \in \Theta} \xi_\sigma(\theta) \left(1 + \bar{A}_{t\theta} - \bar{c}_t\right) = 1 - \bar{c}_t + \sum_{\theta=1}^\nvar \xi_\sigma(\theta) \bar{A}_{t \theta}.
	\end{align*}
	Let $\calT' = \{ t \in [\nequ] \mid \sum_{\theta = 1}^\nvar \xi_\sigma(\theta) \bar{A}_{t \theta} = \bar{c}_t \}$ be the indices of the equations such that the receiver's expected utility values of $(t,a_0)$, $(t,a_1)$, and $(t,a_2)$ are equal to $1$. 
	For $t \in [\nequ] \setminus \calT'$, while the receiver's expected utility of $(t,a_0)$ is $1$, one of the receiver's expected utility values of $(t,a_1)$ or $(t,a_2)$ is greater than $1$ and the other is smaller than $1$.
	Thus, $S^*$ contains $\nequ - |\calT'|$ elements whose expected utility is greater than $1$ and $|\calT'|$ elements whose expected utility equal to $1$.
	Recall that $S^*$ contains at least $\delta \nequ$ elements $(t, a)$ such that $a = a_0$, whose expected utility is $1$.
	Hence, we have $\delta \nequ \le |\calT'|$.
	From the definition of $\calT'$, if we define $\bar{x} \in \bbQ^\nvar$ so that $\bar{x}_\theta = \xi_\sigma(\theta)$ holds for each $\theta \in [\nvar]$, at least $\delta \nequ$ equations in $\bar{A} \bar{x} = \bar{c}$ hold, hence a contradiction.
\end{proof}

\subsection{Proof for Graphic Matroid Constraints}\label{sec:hardness-graphic}
We prove the hardness result for graphic matroid constraints.
We use a reduction similar to the case of uniform matroids, but we need to construct an appropriate graph for obtaining the reduction. 

\thmhardnessgraphic*

\begin{proof}
	As with the case of uniform matroids, we provide a reduction from the binary version of $\textsc{LINEQ-MA}(1-\zeta,\delta)$ with constants $\zeta, \delta \in [0, 1]$ such that $\alpha(1 - 2\zeta) > \delta$. 

	\paragraph{Reduction.}
	We define $\bar{A}$, $\bar{c}$, $\tau$, $\Theta$, and $\mu$ as with the case of uniform matroids.
	We consider an undirected graph with vertices $V = \bigcup_{t = 1}^\nequ \{ v_{t,0}, v_{t,1}, v_{t,2}, v_{t,3}  \}$ and edges $E = \{ \{ v_{t,i}, v_{t,j} \} \mid t \in [\nequ], ~ i,j \in \{0,1,2,3\}, ~ i \neq j \}$, that is, the disjoint union of $\nequ$ complete graphs $K_4$ (see \Cref{fig:graphic-hardness}).
	Given this graph, the receiver can select any spanning tree of $K_4$ for each $t \in [\nequ]$ under the graphic matroid constraint. 
	We define weight $w_\theta(t,i,j)$ of edge $\{v_{t,i}, v_{t,j}\} \in E$ as
	\begin{align*}
		&w_0(t,0,1) = 1, & & w_\theta(t,0,1) = 1,\\
		&w_0(t,0,2) = w_0(t,1,3) = 1 + \bar{c}_t, & & w_\theta(t,0,2) = w_\theta(t,1,3) = 1 - \bar{A}_{t \theta} + \bar{c}_t,\\
		&w_0(t,0,3) = w_0(t,2,3) = 1 - \bar{c}_t, & & w_\theta(t,0,3) = w_\theta(t,2,3) = 1 + \bar{A}_{t \theta}- \bar{c}_t,\\
		& w_0(t,1,2) = L, & & w_\theta(t,1,2) = L
	\end{align*}
	for each $\theta \in [\nvar]$, where $L$ is a sufficiently large value.
	Using these weights, we define the receiver's utility as 
	\[
		r_{\theta}(S) = \sum_{ \{v_{t,i}, v_{t,j}\} \in S} w_\theta(t,i,j).
	\]
	Define the sender's utility as $s_\theta(S) = \frac{1}{\nequ} |\{ t \in [\nequ] \mid \{v_{t,0}, v_{t,1}\} \in S\}|$.

	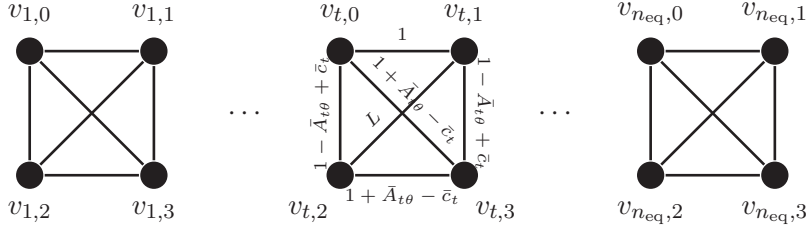
\begin{figure}
	\centering
	\begin{tikzpicture}[line width=1pt]
		\newlength{\kw}
		\setlength{\kw}{0.1\hsize}
		\node [circle, fill=black] at (0, \kw) (v10) {};
		\node [circle, fill=black] at (\kw, \kw) (v11) {};
		\node [circle, fill=black] at (0, 0) (v12) {};
		\node [circle, fill=black] at (\kw, 0) (v13) {};
		\node [circle, fill=black, xshift=0.25\hsize] at (0, \kw) (v20) {};
		\node [circle, fill=black, xshift=0.25\hsize] at (\kw, \kw) (v21) {};
		\node [circle, fill=black, xshift=0.25\hsize] at (0, 0) (v22) {};
		\node [circle, fill=black, xshift=0.25\hsize] at (\kw, 0) (v23) {};
		\node [circle, fill=black, xshift=0.5\hsize] at (0, \kw) (v30) {};
		\node [circle, fill=black, xshift=0.5\hsize] at (\kw, \kw) (v31) {};
		\node [circle, fill=black, xshift=0.5\hsize] at (0, 0) (v32) {};
		\node [circle, fill=black, xshift=0.5\hsize] at (\kw, 0) (v33) {};
		\node [anchor=south] at (v10.north) {$v_{1,0}$};
		\node [anchor=south] at (v11.north) {$v_{1,1}$};
		\node [anchor=north] at (v12.south) {$v_{1,2}$};
		\node [anchor=north] at (v13.south) {$v_{1,3}$};
		\node [anchor=south] at (v20.north) {$v_{t,0}$};
		\node [anchor=south] at (v21.north) {$v_{t,1}$};
		\node [anchor=north east] at (v22.south) {$v_{t,2}$};
		\node [anchor=north west] at (v23.south) {$v_{t,3}$};
		\node [anchor=south] at (v30.north) {$v_{\nequ,0}$};
		\node [anchor=south] at (v31.north) {$v_{\nequ,1}$};
		\node [anchor=north] at (v32.south) {$v_{\nequ,2}$};
		\node [anchor=north] at (v33.south) {$v_{\nequ,3}$};
		\draw [] (v10) -- (v11);
		\draw [] (v10) -- (v12);
		\draw [] (v10) -- (v13);
		\draw [] (v11) -- (v12);
		\draw [] (v11) -- (v13);
		\draw [] (v12) -- (v13);
		\draw [] (v20) to node [above] {\fontsize{7}{7}\selectfont$1$} (v21);
		\draw [] (v20) to node [above, rotate=+90] {\fontsize{7}{7}\selectfont$1 - \bar{A}_{t\theta} + \bar{c}_t$} (v22);
		\draw [] (v20) to node [above, rotate=-45] {\fontsize{7}{7}\selectfont$1 + \bar{A}_{t\theta} - \bar{c}_t$} (v23);
		\draw [] (v21) to node [above, rotate=+45, xshift=-0.2\kw] {\fontsize{7}{7}\selectfont$L$} (v22);
		\draw [] (v21) to node [above, rotate=-90] {\fontsize{7}{7}\selectfont$1 - \bar{A}_{t\theta} + \bar{c}_t$} (v23);
		\draw [] (v22) to node [below] {\fontsize{7}{7}\selectfont$1 + \bar{A}_{t\theta} - \bar{c}_t$} (v23);
		\draw [] (v30) -- (v31);
		\draw [] (v30) -- (v32);
		\draw [] (v30) -- (v33);
		\draw [] (v31) -- (v32);
		\draw [] (v31) -- (v33);
		\draw [] (v32) -- (v33);
		\node[dotted] at ($0.25*(v11)+0.25*(v12)+0.25*(v20)+0.25*(v23)$) {$\cdots$};
		\node[dotted] at ($0.25*(v21)+0.25*(v22)+0.25*(v30)+0.25*(v33)$) {$\cdots$};
	\end{tikzpicture}
	\caption{The undirected graph used to prove the hardness result for Bayesian persuasion with a graphic matroid constraint.}\label{fig:graphic-hardness}
	\end{figure}

	\paragraph{Completeness.} 
	Suppose that there exists a vector $\hat{x} \in \{0,1\}^\nvar$ satisfying at least a $1-\zeta$ fraction of $A \hat{x} = c$ and let $\bar{x} = \frac{1}{\tau} \hat{x}$. 	
	We consider the same signaling scheme as the case of uniform matroids.
	When $\sigma_1$ is observed, for each $t \in [\nequ]$, the expected weight of each edge for the receiver with respect to the posterior distribution $\xi_{\sigma_1} \in \Delta_\Theta$ can be calculated as
	\begin{align*}
		\sum_{\theta \in \Theta} \xi_{\sigma_1}(\theta) w_\theta(t,0,1) &= 1,\\
		\sum_{\theta \in \Theta} \xi_{\sigma_1}(\theta) w_\theta(t,0,2) = \sum_{\theta \in \Theta} \xi_{\sigma_1}(\theta) w_\theta(t,1,3) &= 1 + \bar{c}_t - (\bar{A} \bar{x})_t,\\
		\sum_{\theta \in \Theta} \xi_{\sigma_1}(\theta) w_\theta(t,0,3) = \sum_{\theta \in \Theta} \xi_{\sigma_1}(\theta) w_\theta(t,2,3) &= 1 - \bar{c}_t + (\bar{A} \bar{x})_t,\\
		\sum_{\theta \in \Theta} \xi_{\sigma_1}(\theta) w_\theta(t,1,2) &= L.
	\end{align*}
	For $t \in [\nequ]$ with $(\bar{A} \bar{x})_t = \bar{c}_t$, the receiver's expected utility of $\{v_{t,i},v_{t,j}\}$ is $1$ if $\{i,j\} \neq \{1,2\}$. 
	For each of such $t$, 
	since ties are broken in favor of the sender, 
	the receiver selects a spanning tree consisting of $\{v_{t,1},v_{t,2}\}$, $\{v_{t,0},v_{t,1}\}$, and another edge. 
	For $t \in [\nequ]$ such that $(\bar{A} \bar{x})_t > \bar{c}_t$, the expected utility values of $\{v_{t,0},v_{t,2}\}$ and $\{v_{t,1},v_{t,3}\}$ are greater than $1$.
	For each of such $t \in [\nequ]$, the receiver selects a spanning tree 
	consisting of $\{v_{t,1},v_{t,2}\}$, $\{v_{t,0},v_{t,2}\}$, and $\{v_{t,1},v_{t,3}\}$.
	Similarly, if $(\bar{A} \bar{x})_t < \bar{c}_t$, then the expected utility values of $\{v_{t,0},v_{t,3}\}$ and $\{v_{t,2},v_{t,3}\}$ are greater than $1$.
	For each of such $t \in [\nequ]$, the receiver selects a spanning tree consisting of $\{v_{t,1},v_{t,2}\}$, $\{v_{t,0},v_{t,3}\}$, and $\{v_{t,2},v_{t,3}\}$. 
	Therefore, the expected utility of the sender is at least $1-2\zeta$ as with the case of uniform matroids, where the second inequality is obtained by letting $\nvar$ be sufficiently large.
	Hence, an $\alpha$-approximation algorithm returns a signaling scheme with an expected utility greater than $\alpha(1-2\zeta)$.

	\paragraph{Soundness.}
	As with the case of uniform matroids, we prove by contradiction that if every vector $\hat{x} \in \bbQ^\nvar$ satisfies less than a $\delta$ fraction of the equations in $A \hat{x} = c$, no algorithm for Bayesian persuasion with a graphic matroid constraint returns a signaling scheme with an expected utility greater than or equal to $\delta$.

	Assume that there exists a signaling scheme $(\phi_\theta)_{\theta \in \Theta}$ that provides the sender with an expected utility at least $\delta$.
	By an averaging argument, there exists a signal $\sigma \in \Sigma$ such that the receiver's best response to $\sigma$ yields the sender's expected utility at least $\delta$.
	Let $S^* \subseteq E$ be this best response, which contains at least $\delta \nequ$ edges represented as $\{v_{t,0}, v_{t,1}\}$ for some $t \in [\nequ]$.
	Let $\xi_\sigma \in \Delta_\Theta$ be the posterior probability distribution when $\sigma$ is observed.
	For each $t \in [\nequ]$, the expected weight of each edge can be expressed as 
	\begin{align*}
		\sum_{\theta \in \Theta} \xi_{\sigma_1}(\theta) w_\theta(t,0,1) &= 1,\\
		\sum_{\theta \in \Theta} \xi_{\sigma_1}(\theta) w_\theta(t,0,2) = \sum_{\theta \in \Theta} \xi_{\sigma_1}(\theta) w_\theta(t,1,3) &= 1 + \bar{c}_t - \sum_{\theta = 1}^\nvar \xi_\sigma(\theta) \bar{A}_{t \theta},\\
		\sum_{\theta \in \Theta} \xi_{\sigma_1}(\theta) w_\theta(t,0,3) = \sum_{\theta \in \Theta} \xi_{\sigma_1}(\theta) w_\theta(t,2,3) &= 1 - \bar{c}_t + \sum_{\theta = 1}^\nvar \xi_\sigma(\theta) \bar{A}_{t \theta},\\
		\sum_{\theta \in \Theta} \xi_{\sigma_1}(\theta) w_\theta(t,1,2) &= L.
	\end{align*}
	Let $\calT' = \{ t \in [\nequ] \mid \sum_{\theta = 1}^\nvar \xi_\sigma(\theta) \bar{A}_{t \theta} = \bar{c}_t \}$ be the set of indices such that $\sum_{\theta = 1}^\nvar \xi_\sigma(\theta) \bar{A}_{t \theta} = \bar{c}_t$ holds.
	For $t \in [\nequ] \setminus \calT'$, the receiver does not select $\{v_{t,0}, v_{t,1}\}$.
	Since $S^*$ contains at least $\delta \nequ$ edges represented as $\{v_{t,0}, v_{t,1}\}$ for some $t \in [\nequ]$, it must hold that $|\calT'| \ge \delta \nequ$.
	If we define $\bar{x} \in \bbQ^\nvar$ such that $\bar{x}_\theta = \xi_\sigma(\theta)$ for each $\theta \in [\nvar]$, then at least $\delta \nequ$ equations in $\bar{A} \bar{x} = \bar{c}$ hold, leading to a contradiction.
\end{proof}

\subsection{Proof for Path Constraints}\label{sec:hardness-path}
We prove the hardness result for path constraints.
We again use almost the same reduction as the one used for uniform matroids, but we need to construct an appropriate graph and deal with the minimization setting.

\thmhardnesspath*

\begin{proof}
	As with the case of uniform matroids, we provide a reduction from $\textsc{LINEQ-MA}(1-\zeta,\delta)$. 
	To deal with the minimization setting, we here set constants $\zeta, \delta \in [0, 1]$ so that $1 - \delta > \alpha \zeta (1 + \zeta)$ holds.

	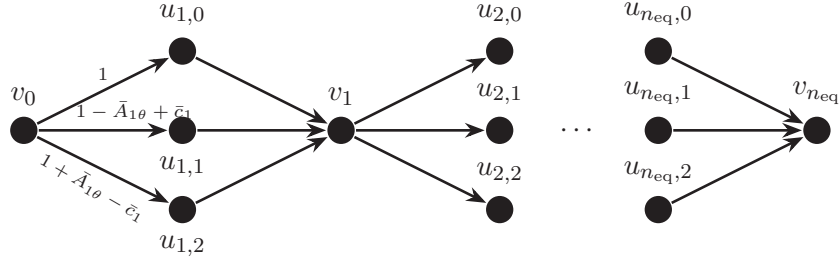
\begin{figure}
	\centering
	\begin{tikzpicture}[line width=1pt]
		\newlength{\el}
		\setlength{\el}{30pt}
		\node [circle, fill=black] at (0\el, 1\el) (v0) {};
		\node [circle, fill=black] at (4\el, 1\el) (v1) {};
		\node [circle, fill=black] at (10\el, 1\el) (v3) {};
		\node [circle, fill=black] at (2\el, 2\el) (u10) {};
		\node [circle, fill=black] at (2\el, 1\el) (u11) {};
		\node [circle, fill=black] at (2\el, 0\el) (u12) {};
		\node [circle, fill=black] at (6\el, 2\el) (u20) {};
		\node [circle, fill=black] at (6\el, 1\el) (u21) {};
		\node [circle, fill=black] at (6\el, 0\el) (u22) {};
		\node [circle, fill=black] at (8\el, 2\el) (u30) {};
		\node [circle, fill=black] at (8\el, 1\el) (u31) {};
		\node [circle, fill=black] at (8\el, 0\el) (u32) {};
		\node [anchor=south] at (v0.north) {$v_0$};
		\node [anchor=south] at (v1.north) {$v_1$};
		\node [anchor=south] at (v3.north) {$v_\nequ$};
		\node [anchor=south] at (u10.north) {$u_{1,0}$};
		\node [anchor=north] at (u11.south) {$u_{1,1}$};
		\node [anchor=north] at (u12.south) {$u_{1,2}$};
		\node [anchor=south] at (u20.north) {$u_{2,0}$};
		\node [anchor=south] at (u21.north) {$u_{2,1}$};
		\node [anchor=south] at (u22.north) {$u_{2,2}$};
		\node [anchor=south] at (u30.north) {$u_{\nequ,0}$};
		\node [anchor=south] at (u31.north) {$u_{\nequ,1}$};
		\node [anchor=south] at (u32.north) {$u_{\nequ,2}$};
		\draw [-{Stealth}] (v0) to node [above] {\fontsize{7}{7}\selectfont$1$} (u10);
		\draw [-{Stealth}] (v0) to node [above, xshift=0.4\el] {\fontsize{7}{7}\selectfont$1 - \bar{A}_{1\theta} + \bar{c}_1$} (u11);
		\draw [-{Stealth}] (v0) to node [below, rotate=-30] {\fontsize{7}{7}\selectfont$1 + \bar{A}_{1\theta} - \bar{c}_1$} (u12);
		\draw [-{Stealth}] (u10) -- (v1);
		\draw [-{Stealth}] (u11) -- (v1);
		\draw [-{Stealth}] (u12) -- (v1);
		\draw [-{Stealth}] (v1) -- (u20);
		\draw [-{Stealth}] (v1) -- (u21);
		\draw [-{Stealth}] (v1) -- (u22);
		\draw [-{Stealth}] (u30) -- (v3);
		\draw [-{Stealth}] (u31) -- (v3);
		\draw [-{Stealth}] (u32) -- (v3);
		\node[dotted] at ($0.5*(u21)+0.5*(u31)$) {$\cdots$};
	\end{tikzpicture}
	\caption{The directed graph used to prove the hardness result for Bayesian persuasion with path constraints}\label{fig:path-hardness}
	\end{figure}

	\paragraph{Reduction.}
	We define $\bar{A}$, $\bar{c}$, $\tau$, $\Theta$, and $\mu$ as with the case of uniform matroids.
	We consider a directed graph with vertices $V = \{ v_0 \} \cup \{ v_t \mid t \in [\nequ] \} \cup \{ u_{t,i} \mid t \in [\nequ], ~ i \in \{0,1,2\} \}$ and edges $E = \{ (v_{t-1}, u_{t,i}) \mid t \in [\nequ], ~ i \in \{0,1,2\} \} \cup \{ (u_{t,i}, v_t) \mid t \in [\nequ], ~ i \in \{0,1,2\} \}$ (see \Cref{fig:path-hardness}).
	Since there are three choices at $v_{t-1}$ for each $t \in [\nequ]$, the number of $v_0$--$v_\nequ$ paths in this graph is $3^\nequ$.
	We define weight $w_\theta(u,v)$ of each edge $(u,v) \in E$ as
	\begin{align*}
		&w_0(v_{t-1}, u_{t,0}) = 1, & & w_\theta(v_{t-1}, u_{t,0}) = 1,\\
		&w_0(v_{t-1}, u_{t,1}) = 1 + \bar{c}_t, & & w_\theta(v_{t-1}, u_{t,1}) = 1 - \bar{A}_{t \theta} + \bar{c}_t,\\
		&w_0(v_{t-1}, u_{t,2}) = 1 - \bar{c}_t, & & w_\theta(v_{t-1}, u_{t,2}) = 1 + \bar{A}_{t \theta}- \bar{c}_t
	\end{align*}
	for each $\theta \in [\nvar]$. 
	We set weights of the other edges to zero, i.e., 
	$w_\theta(u_{t,i}, v_t) = 0$ for each $t \in [\nequ]$, $i \in \{0,1,2\}$, and $\theta \in \Theta$.
	Using these weights, we define the receiver's cost of path $S \subseteq E$ as $r_{\theta}(S) = \sum_{ (u,v) \in S} w_\theta(u,v)$.
	Define the sender's cost as 
	\[
		s_\theta(S) = \frac{1}{\nequ} |\{ (v_t, u_{t,i}) \in S \mid t \in [\nequ], ~ i \in \{1,2\} \}|. 
	\]

	\paragraph{Completeness.}
	Suppose $\hat{x} \in \{0,1\}^\nvar$ to be a vector satisfying at least a $1-\zeta$ fraction of $A \hat{x} = c$ and let $\bar{x} = \frac{1}{\tau} \hat{x}$; therefore, at most a $\zeta$ fraction of $\bar{A}\bar{x} = \bar{c}$ can be violated. 
	We consider the same signaling scheme as the case of uniform matroids.
	When $\sigma_1$ is observed, for each $t \in [\nequ]$, the expected weight of the edges for the receiver can be calculated as
	\begin{align*}
		&\sum_{\theta \in \Theta} \xi_{\sigma_1}(\theta) w_\theta(v_{t-1}, u_{t,0}) = 1,\\
		&\sum_{\theta \in \Theta} \xi_{\sigma_1}(\theta) w_\theta(v_{t-1}, u_{t,1}) = 1 + \bar{c}_t - (\bar{A} \bar{x})_t,\\
		&\sum_{\theta \in \Theta} \xi_{\sigma_1}(\theta) w_\theta(v_{t-1}, u_{t,2}) = 1 - \bar{c}_t + (\bar{A} \bar{x})_t.
	\end{align*}
	For $t \in [\nequ]$ with $(\bar{A} \bar{x})_t = \bar{c}_t$, the receiver's expected cost is identical for all of $(v_{t-1},u_{t,0})$, $(v_{t-1},u_{t,1})$, and $(v_{t-1},u_{t,2})$.
	Since we can break ties in favor of the sender, we can assume that the receiver selects $(v_{t-1},u_{t,0})$ and $(u_{t,0}, v_t)$ for such $t$.
	On the other hand, for $t \in [\nequ]$ such that $(\bar{A} \bar{x})_t \neq \bar{c}_t$, the receiver selects either $(v_{t-1},u_{t,1})$ or $(v_{t-1},u_{t,2})$.
	Therefore, the expected cost of the sender is at most
	\begin{equation*}
		\left( \sum_{\theta \in \Theta} \mu(\theta) \phi_{\theta}(\sigma_1) \right) \frac{1}{\nequ} \zeta \nequ + \left( \sum_{\theta \in \Theta} \mu(\theta) \phi_{\theta}(\sigma_2) \right) \frac{1}{\nequ} \nequ \le \zeta + \frac{1}{\nvar} \le \zeta (1 + \zeta),
	\end{equation*}
	where the second inequality is obtained by letting $\nvar$ be sufficiently large. 
	Therefore, an $\alpha$-approximation algorithm returns a signaling scheme with an expected cost at most $\alpha \zeta (1 + \zeta)$.

	\paragraph{Soundness.}
	As with the case of uniform matroids, we prove by contradiction that if every vector $\hat{x} \in \bbQ^\nvar$ satisfies less than a $\delta$ fraction of the equations in $A \hat{x} = c$, no algorithm for Bayesian persuasion with path constraints returns a signaling scheme with an expected cost at most $1-\delta$.

	Suppose for contradiction that there exists a signaling scheme $(\phi_\theta)_{\theta \in \Theta}$ with which the sender incurs an expected cost at most $1-\delta$.
	By an averaging argument, there exists a signal $\sigma \in \Sigma$ such that the receiver's best response to $\sigma$ yields the sender's expected cost at most $1-\delta$.
	Let $S^* \subseteq E$ be this best response.
	Then, $S^*$ contains at most $(1 - \delta) \nequ$ edges represented as $(v_{t-1},u_{t,1})$ or $(v_{t-1},u_{t,2})$ for some $t \in [\nequ]$. 
	In other words, $S^*$ contains at least $\delta \nequ$ edges represented as $(v_{t-1},u_{t,0})$ for some $t \in [\nequ]$.
	Let $\xi_\sigma \in \Delta_\Theta$ be the posterior probability distribution when $\sigma$ is observed.
	For each $t \in [\nequ]$, the expected weight of each edge is
	\begin{align*}
		&\sum_{\theta \in \Theta} \xi_{\sigma_1}(\theta) w_\theta(v_{t-1}, u_{t,0}) = 1,\\
		&\sum_{\theta \in \Theta} \xi_{\sigma_1}(\theta) w_\theta(v_{t-1}, u_{t,1}) = 1 + \bar{c}_t - \sum_{\theta=1}^\nvar \xi_\sigma(\theta) \bar{A}_{t \theta},\\
		&\sum_{\theta \in \Theta} \xi_{\sigma_1}(\theta) w_\theta(v_{t-1}, u_{t,2}) = 1 - \bar{c}_t + \sum_{\theta=1}^\nvar \xi_\sigma(\theta) \bar{A}_{t \theta}.
	\end{align*}
	Let $\calT' = \{ t \in [\nequ] \mid \sum_{\theta = 1}^\nvar \xi_\sigma(\theta) \bar{A}_{t \theta} = \bar{c}_t \}$ be the set of indices such that $\sum_{\theta = 1}^\nvar \xi_\sigma(\theta) \bar{A}_{t \theta} = \bar{c}_t$ holds.
	For $t \in [\nequ] \setminus \calT'$, the receiver does not select $(v_{t-1}, u_{t,0})$.
	Since $S^*$ contains at least $\delta \nequ$ edges represented as $(v_{t-1},u_{t,0})$ for some $t \in [\nequ]$, it must hold that $|\calT'| \ge \delta \nequ$.
	If we define $\bar{x} \in \bbQ^\nvar$ such that $\bar{x}_\theta = \xi_\sigma(\theta)$ for each $\theta \in [\nvar]$, then at least $\delta \nequ$ equations in $\bar{A} \bar{x} = \bar{c}$ hold, hence a contradiction.
\end{proof}

\section{Omitted Details of Polynomial-time Algorithm for Constant Number of States}

\subsection{Proofs of the Equivalence of the LP Formulations}\label{sec:reduced-lp}

In this section, we prove the following proposition.

\propreducedlp*

We first show that if $S \in \calI$ is not one of the receiver's best responses for any posterior distribution $\xi \in \Delta_\Theta$, variable $\phi_\theta(S)$ never takes a non-zero value.

\begin{proposition}\label{prop:reduced-variables}
	Any feasible solution $(\phi_\theta)_{\theta \in \Theta}$ of \eqref{eq:lp} satisfies $\phi_\theta(S) = 0$ for all $S \in \calI \setminus \calI^*$ and $\theta \in \Theta$.
\end{proposition}

\begin{proof}
	We prove the contraposition: if $\phi_\theta(S) > 0$ for some $\theta \in \Theta$, then $S \in \calI^*$.
	Let $(\phi_\theta)_{\theta \in \Theta}$ be an arbitrary feasible solution to \eqref{eq:lp}.
	Fix $S \in \calI$ such that $\phi_\theta(S) > 0$ for some $\theta \in \Theta$.
	Since we assume $\mu(\theta) > 0$ for all $\theta \in \Theta$, we have $\sum_{\theta \in \Theta} \mu(\theta) \phi_{\theta}(S) > 0$.
	Since $(\phi_\theta)_{\theta \in \Theta}$ satisfies $\sum_{\theta \in \Theta} \mu(\theta) \phi_\theta(S) (r_\theta(S) - r_\theta(S')) \ge 0$ for every $S' \in \calI$ due to the first constraint of \eqref{eq:lp}, we have
	\begin{equation*}
		\sum_{\theta \in \Theta} \frac{\mu(\theta) \phi_\theta(S)}{\sum_{\theta' \in \Theta} \mu(\theta') \phi_{\theta'}(S)} r_\theta(S) \ge \sum_{\theta \in \Theta} \frac{\mu(\theta) \phi_\theta(S)}{\sum_{\theta' \in \Theta} \mu(\theta') \phi_{\theta'}(S)} r_\theta(S')
	\end{equation*}
	for every $S' \in \calI$.
	By setting $\xi(\theta) = \frac{\mu(\theta) \phi_\theta(S)}{\sum_{\theta' \in \Theta} \mu(\theta') \phi_{\theta'}(S)}$ for each $\theta \in \Theta$, we have
	\begin{equation*}
		\sum_{\theta \in \Theta} \xi(\theta) r_\theta(S) \ge \sum_{\theta \in \Theta} \xi(\theta) r_\theta(S')
	\end{equation*}
	for every $S' \in \calI$.
	Note that $\xi$ is a probability distribution over $\Theta$, that is, $\xi \in \Delta_\Theta$.
	Therefore, $S$ is the receiver's best response for this $\xi$, implying $S \in \calI^*$.
\end{proof}

We then show that considering the first constraint $\sum_{\theta \in \Theta} \mu(\theta) \phi_\theta(S) \left( r_\theta(S) - r_\theta(S') \right) \ge 0$ only for $S' \in \calI^*$ is sufficient. 
That is, although the reduced formulation \eqref{eq:reduced-lp} has the first constraint only for $S' \in \calI^*$, any feasible solution to \eqref{eq:reduced-lp} also satisfies it for all $S' \in \calI$.

\begin{proposition}\label{prop:reduced-constraints}
	Any feasible solution $(\phi_\theta)_{\theta \in \Theta}$ for \eqref{eq:reduced-lp} satisfies $\sum_{\theta \in \Theta} \mu(\theta) \phi_\theta(S) \left( r_\theta(S) - r_\theta(S') \right) \ge 0$ for all $S \in \calI^*$ and $S' \in \calI$.
\end{proposition}

\begin{proof}
	Fix any feasible solution $(\phi_\theta)_{\theta \in \Theta}$ to \eqref{eq:reduced-lp} and any $S \in \calI^*$.
	From the first constraint of \eqref{eq:reduced-lp}, it holds that 
	\[
		\sum_{\theta \in \Theta} \mu(\theta) \phi_\theta(S) \left( r_\theta(S) - r_\theta(S') \right) \ge 0
	\]
	for all $S' \in \calI^*$.
	By setting $\xi(\theta) = \frac{\mu(\theta) \phi_\theta(S)}{\sum_{\theta' \in \Theta} \mu(\theta') \phi_{\theta'}(S)}$ for each $\theta \in \Theta$, we have
	\begin{equation}\label{eq:br-ineq}
		\sum_{\theta \in \Theta} \xi(\theta) r_\theta(S) \ge \sum_{\theta \in \Theta} \xi(\theta) r_\theta(S')
	\end{equation}
	for all $S' \in \calI^*$.
	Since $\xi \in \Delta_\Theta$, if we let $S^* \in \argmax_{S \in \calI} \sum_{\theta \in \Theta} \xi(\theta) r_\theta(S)$, then we have $S^* \in \calI^*$.
	For any $S' \in \calI$, we have
	\begin{equation*}
		\sum_{\theta \in \Theta} \xi(\theta) r_\theta(S) \ge \sum_{\theta \in \Theta} \xi(\theta) r_\theta(S^*) \ge \sum_{\theta \in \Theta} \xi(\theta) r_\theta(S'),
	\end{equation*}
	where the first inequality is due to $S^* \in \calI^*$ and \eqref{eq:br-ineq}, and the second inequality is due to the definition of $S^*$.
\end{proof}

From \Cref{prop:reduced-variables,prop:reduced-constraints}, any feasible solution to \eqref{eq:lp} can be converted to a feasible solution to \eqref{eq:reduced-lp} by removing variables corresponding to $S \in \calI \setminus \calI^*$, which are all zeros, and any feasible solution to \eqref{eq:reduced-lp} can be converted to a feasible solution to \eqref{eq:lp} by adding zero variables corresponding to $S \in \calI \setminus \calI^*$.
The latter is because thus obtained solution $(\phi_\theta)_{\theta \in \Theta}$ satisfies the first constraint of \eqref{eq:lp} for $S \in \calI \setminus \calI^*$ by making the left-hand side zero.
Moreover, since the removed or added variables are all zeros, this conversion does not change the objective value.
Therefore, \Cref{prop:reduced-lp} holds.

\subsection{Proofs for Enumeration of the Best Responses}\label{sec:degeneracy}

First, we prove the following lemma.

\lemdegeneracy*

\begin{proof}
	Fix any permutation $\pi$ of $E$.
	Let $\xi \in \bbR^\Theta$ be a vector such that $\psi_{\pi(1)}^\top \xi \ge \cdots \ge \psi_{\pi(n)}^\top \xi$ and $\bfone^\top \xi = 1$. 
	If $\xi$ also satisfies $\psi_{\pi(1)}^\top \xi > \cdots > \psi_{\pi(n)}^\top \xi$, then $\{ \xi \in \bbR^\Theta \mid \psi_{\pi(1)}^\top \xi > \cdots > \psi_{\pi(n)}^\top \xi ~\text{and}~ \bfone^\top \xi = 1\} \neq \emptyset$ holds. 
	Otherwise, $\psi_{\pi(i)}^\top \xi = \psi_{\pi(i+1)}^\top \xi$ holds for some $i \in [n-1]$. 
	Below we show that these ties can be broken by perturbing $\xi$ while keeping the sum equal to $1$. 

	Let $T \subseteq [n-1]$ be the set of indices such that 
	$\psi_{\pi(i)}^\top \xi = \psi_{\pi(i+1)}^\top \xi$ holds if and only if $i \in T$. 	
	We prove $|T| \le |\Theta| - 1$.
	Assume for contradiction that $|T| \ge |\Theta|$ holds.
	Let $T' \subseteq T$ be any subset such that $|T'| = |\Theta|$.
	From \Cref{asm:degeneracy}, vectors $\{ \psi_{\pi(i)} - \psi_{\pi(i+1)} \mid i \in T' \}$ are linearly independent.
	Thus, a matrix whose columns are given by $(\psi_{\pi(i)} - \psi_{\pi(i+1)})^\top$ for $i \in T'$ is regular. 
	Since $\xi$ satisfies $(\psi_{\pi(i)} - \psi_{\pi(i+1)})^\top \xi = 0$ for all $i \in T'$, we have $\xi = \bfzero$, which contradicts $\bfone^\top \xi = 1$.
	Hence, we have $|T| \le |\Theta| - 1$.

	We then consider a linear system that consists of $(\psi_{\pi(i)} - \psi_{\pi(i+1)})^\top \eta = 1$ for all $i \in T$ and $\bfone^\top \eta = 0$.
	Since this linear system has $|\Theta|$ variables and at most $|\Theta|$ equations, there exists a solution $\eta \in \bbR^\Theta$ that satisfies all these equations.
	Hence, for sufficiently small $\delta > 0$, it holds that $\psi_{\pi(1)}^\top (\xi + \delta \eta) > \cdots > \psi_{\pi(n)}^\top (\xi + \delta \eta)$ and $\bfone^\top (\xi + \delta \eta) = 1$. Therefore, we obtain $\{ \xi \in \bbR^\Theta \mid \psi_{\pi(1)}^\top \xi > \cdots > \psi_{\pi(n)}^\top \xi ~\text{and}~ \bfone^\top \xi = 1\} \neq \emptyset$.
\end{proof}

By using \Cref{lem:degeneracy}, we prove the following theorem.

\lemgeneralenumeration*

\begin{proof}
	First, we prove $\calI^* \subseteq \calI_\calC$.
	Fix $S \in \calI^*$.
	From the definition of $\calI^*$, $S$ is a maximum weight independent set with respect to weights $\sum_{\theta \in \Theta} \xi(\theta) r_\theta(\{\cdot\})$ for some $\xi \in \Delta_\Theta$.
	Then, from the property of matroids, there exists a permutation $\pi$ of $E$ such that (i) the greedy algorithm returns $S$ by examining the elements in the order of $\pi$ and (ii) $\sum_{\theta \in \Theta} \xi(\theta) r_\theta(\{\pi(1)\}) \ge \cdots \ge \sum_{\theta \in \Theta} \xi(\theta) r_\theta(\{\pi(n)\})$.

	Let $C \in \calC$ be a cell in which $\sum_{\theta \in \Theta} \xi(\theta) r_\theta(\{\pi(1)\}) \ge \cdots \ge \sum_{\theta \in \Theta} \xi(\theta) r_\theta(\{\pi(n)\})$ holds.
	Even if $\xi$ is not an interior point of $C$, under \Cref{asm:degeneracy}, \Cref{lem:degeneracy} implies that there exists some $\tilde{\xi} \in \aff(\Delta_\Theta)$ that is an interior point of $C$.
	Since the output of the greedy algorithm depends only on the order of the weights, $S$ is also a maximum weight independent set for $\sum_{\theta \in \Theta} \tilde{\xi}(\theta) r_\theta(\{\cdot\})$.

	From the construction of $\calI_\calC$, there exists a maximum weight independent set $S' \in \calI_\calC$ for an interior point $\xi'$ of $C$.
	Since both $\tilde{\xi}$ and $\xi'$ are interior points of $C$, the signs of $\sum_{\theta \in \Theta} \tilde{\xi}(\theta) (r_\theta(\{i\}) - r_\theta(\{j\}))$ and $\sum_{\theta \in \Theta} \xi'(\theta) (r_\theta(\{i\}) - r_\theta(\{j\}))$ are identical for every pair $i, j \in E$ with $i \neq j$. 
	Thus, the descending orders of the weights $\sum_{\theta \in \Theta} \tilde{\xi}(\theta) r_\theta(\{\cdot\})$ and $\sum_{\theta \in \Theta} \xi'(\theta) r_\theta(\{\cdot\})$ are also identical.
	Since the output of the greedy algorithm depends only on the descending order of the weights, we obtain $S = S'$.
	Hence, $S \in \calI_\calC$, which leads to $\calI^* \subseteq \calI_\calC$.

	Next, we provide an upper bound on $|\calI_\calC|$.
	The set $\calH$ contains one hyperplane for each pair distinct $i,j \in E$.
	Hence $|\calH| = O(n^2)$.
	From the well-known fact that $n$ hyperplanes decompose $\bbR^d$ into $O(n^d)$ cells (see, e.g., \citet{Ede87}), we obtain $|\calI_\calC| = O(n^{2d})$.
\end{proof}

\section{Omitted Details of Polynomial-time Algorithm for CCE-Persuasiveness}\label{sec:app-cce}

\subsection{Proof for Approximate Separation Oracles}\label{sec:cce-separation}

We prove the following lemma.

\lemseparation*

\begin{proof}
	Since the separation problem for $y\ge0$ is trivial, we below suppose $y\ge0$ to hold. 
	If the $\alpha$-approximation oracle returns $S'\in\calI$ such that $x_\theta < \mu(\theta) (s_\theta(S') + r_\theta(S') \cdot y)$ for some $\theta \in \Theta$, then $x_\theta - \mu(\theta) r_\theta(S') \cdot y = \mu(\theta) s_\theta(S')$ is a separating hyperplane. 
	Otherwise, from the definition of the $\alpha$-approximation oracle, 
	$\{x_\theta\}_{\theta\in\Theta}$ and $y$ are guaranteed to satisfy $x_\theta \ge \mu(\theta) (\alpha\cdot s_\theta(S) + r_\theta(S) \cdot y)$ for every $S\in\calI$, which implies the feasibility of $\{x_\theta / \alpha\}_{\theta\in\Theta}$ and $y/\alpha$. 
\end{proof}

\subsection{On the Parameters of the Binary Search}\label{sec:cce-parameters}

As claimed in \Cref{thm:cce-polytime}, \Cref{alg:cce} needs $\vmin$ and $\vmax$ as inputs such that (i) $\opt > 0$ implies $\opt > \vmin$ and (ii) $\opt < \vmax$, where $\opt$ is the optimal value.
We describe how to compute $\vmin$ and $\vmax$ for \Cref{alg:cce} when $s_\theta$ is an integer-valued monotone submodular function and $r_\theta$ is an integer-valued gross substitute function for each $\theta \in \Theta$.
We can deal with special cases where $s_\theta$ or $r_\theta$ is linear in the same way. 

First, we consider how to set $\vmax$, which must upper-bound $\opt$. 
Let $(\phi_\theta)_{\theta \in \Theta}$ be any signaling scheme.
Since the sender's expected utility is a linear combination of the sender's utility for each $\theta$ and $S$, we have
\begin{equation*}
	\sum_{\theta \in \Theta} \sum_{S\in \calI} \mu(\theta) \phi_\theta(S) s_\theta(S)
	\le \max_{\theta \in \Theta} \max_{S \in \calI} s_\theta(S)
	\le |E| \max_{\theta \in \Theta} \max_{i \in E} s_\theta(\{i\}),
\end{equation*}
where the second inequality is due to submodularity of $s_\theta$.
Thus, by setting $\vmax > |E| \max_{\theta \in \Theta} \max_{i \in E} s_\theta(\{i\})$, 
we obtain $\opt < \vmax$.

Next, we consider how to set $\vmin$, which must lower-bound $\opt$ whenever $\opt>0$. 
The optimal value of LP \eqref{eq:cce-lp} can be expressed as $\sum_{\theta \in \Theta} \sum_{S \in \calI} \mu(\theta) \phi_\theta(S) s_\theta(S)$, where $(\phi_\theta)_{\theta \in \Theta}$ is an optimal solution that corresponds to a vertex of the feasible region of LP \eqref{eq:cce-lp}.
If the optimal value is non-zero, there exists $\theta^* \in \Theta$ and $S^* \in \calI$ such that $\mu(\theta^*) \phi_{\theta^*}(S^*) s_{\theta^*}(S^*) > 0$, which lower-bounds the optimal value. 

In the following, we lower-bound each of $\mu(\theta^*)$, $\phi_{\theta^*}(S^*)$, and $s_{\theta^*}(S^*)$. 
The first and last are bounded as 
\begin{align*}
	\mu(\theta^*) \ge \min_{\theta \in \Theta \colon \mu(\theta) \neq 0} \mu(\theta)
	& & 
	\text{and} 
	& &
	s_{\theta^*}(S^*) \ge \min_{\theta \in \Theta} \min_{i \in E \colon s_{\theta}(\{i\}) \neq 0} s_{\theta}(\{i\}),
\end{align*}
respectively, 
where the latter is due to monotonicity of $s_\theta$.
We discuss how to lower-bound $\phi_{\theta^*}(S^*)$.
The LP \eqref{eq:cce-lp} has $|\Theta|$ equality constraints and $1 + |\Theta| |\calI|$ inequality constraints, while the number of variables is $|\Theta| |\calI|$.
Therefore, each vertex of the polytope is characterized by $|\Theta|$ equality constants and $|\Theta| |\calI| - |\Theta|$ active inequality constraints. 
Note that $1 + |\Theta| |\calI|$ inequality constraints consist of $1$ CCE-persuasiveness constraint and $|\Theta| |\calI|$ non-negativity constraints.
Since $(\phi_\theta)_{\theta \in \Theta}$ is a vertex of the polytope, it is sufficient to lower-bound non-zero entries of every vertex of the polytope, which yields a lower bound on $\phi_{\theta^*}(S^*)$. 

If constraint $\sum_{\theta \in \Theta} \sum_{S\in \calI} \mu(\theta) \phi_\theta(S) r_\theta(S) \ge C$ is inactive (i.e., the strict inequality holds) at vertex $(\phi_\theta)_{\theta \in \Theta}$, then $|\Theta| |\calI| - |\Theta|$ non-negativity constraints must be active.
Since equality constraints $\sum_{S \in \calI} \phi_\theta(S) = 1$ also hold for all $\theta \in \Theta$, each $\phi_\theta$ is a vertex of the probability simplex, i.e., for each $\theta \in \Theta$, $\phi_\theta(S) = 1$ for some $S \in \calI$ and $\phi_\theta(S') = 0$ for the other $S' \in \calI$.
Therefore, all non-zeros are equal to $1$. 

Suppose that $\sum_{\theta \in \Theta} \sum_{S\in \calI} \mu(\theta) \phi_\theta(S) r_\theta(S) = C$ holds at vertex $(\phi_\theta)_{\theta \in \Theta}$. 
In this case, for a single $\theta'\in\Theta$, 
there may exist two variables $\phi_{\theta'}(S_1), \phi_{\theta'}(S_2) \in (0, 1)$ such that 
$\phi_{\theta'}(S_1) + \phi_{\theta'}(S_2) = 1$.  
For the other $\theta\in\Theta$, $\phi_\theta$ is a vertex of the probability simplex, 
i.e., $\phi_\theta(S_\theta) = 1$ for some $S_\theta\in\calI$ and $\phi_\theta(S) = 0$ for the other $S\in\calI$. 
By substituting these into $\sum_{\theta \in \Theta} \sum_{S\in \calI} \mu(\theta) \phi_\theta(S) r_\theta(S) = C$, we obtain
\begin{equation*}
	\sum_{\theta \in \Theta \colon \theta \neq \theta'} \mu(\theta) r_\theta(S_\theta) + \mu(\theta') \phi_{\theta'}(S_1) r_{\theta'}(S_1) + \mu(\theta') \phi_{\theta'}(S_2) r_{\theta'}(S_2)= C.
\end{equation*}
By considering $\phi_{\theta'}(S_1) + \phi_{\theta'}(S_2) = 1$, we can solve it for $\phi_{\theta'}(S_1)$ as
\begin{equation*}
	\phi_{\theta'}(S_1) = \frac{C - \sum_{\theta \in \Theta \colon \theta \neq \theta'} \mu(\theta) r_{\theta}(S_{\theta}) - \mu(\theta') r_{\theta'}(S_2)}{\mu(\theta') (r_{\theta'}(S_1) - r_{\theta'}(S_2))}
	= \frac{\left| C - \sum_{\theta \in \Theta \colon \theta \neq \theta'} \mu(\theta) r_{\theta}(S_{\theta}) - \mu(\theta') r_{\theta'}(S_2) \right|}{\mu(\theta')|r_{\theta'}(S_1) - r_{\theta'}(S_2)|},    
\end{equation*}
and we can solve for $\phi_{\theta'}(S_2)$ in the same way 
(if $r_{\theta'}(S_1) - r_{\theta'}(S_2) = 0$, the vertex $(\phi_\theta)_{\theta \in \Theta}$ is again a $0$-$1$ vector as with the previous case). 
The denominator is at most
\begin{equation*}
	\mu(\theta') |r_{\theta'}(S_1) - r_{\theta'}(S_2)| \le |E| \max_{\tilde{\theta} \in \Theta} \max_{i \in E} r_{\tilde{\theta}}(\{i\})
\end{equation*}
due to submodularity of $r_{\tilde{\theta}}$ (every gross substitute function has submodularity). 
Recalling $C$ is defined as $C = \max_{S' \in \calI} \sum_{\theta \in \Theta} \mu(\theta) r_\theta(S')$, the numerator can be written as $\sum_{\theta \in \Theta} z_\theta \mu(\theta)$ for some integers $z_\theta$; 
therefore, it can be bounded from below by the highest precision for the bit representation of $\mu$.
Thus, the numerator and denominator are represented in polynomial space, and so is $\phi_\theta^*(S^*)$. 
To conclude, if $\opt>0$, we can set $\vmin$ so that it can be represented in polynomial space and $\vmin < \mu(\theta^*) \phi_{\theta^*}(S^*) s_{\theta^*}(S^*)\le\opt$ holds. 
If we set $\vmax$ and $\vmin$ as above, the number of iterations of \Cref{alg:cce}, $O\big(\log\frac{\vmax}{\epsilon\vmin}\big)$, is polynomial.

\end{document}